\documentclass[a4paper,11pt]{article}

\usepackage[labelfont=bf,labelsep=period]{caption}
\usepackage[nodayofweek]{datetime}
\usepackage{enumitem}
\usepackage{float}
\usepackage[margin=2.5cm]{geometry}
\usepackage{graphicx}
\usepackage{hyperref}
\usepackage{numbertabbing}
\usepackage{times}
\usepackage{url}
\usepackage{xspace}

\usepackage{amsmath} 
\allowdisplaybreaks[2]          
\usepackage{amssymb} 
\usepackage{amsthm} 
\newtheoremstyle{plain-boldhead}
  {\topsep}
  {\topsep}
  {\itshape}
  {}
  {\bfseries}
  {.}
  { }
  {\thmname{#1}\thmnumber{ #2}\thmnote{ (\bfseries #3)}}
\newtheoremstyle{definition-boldhead}
  {\topsep}
  {\topsep}
  {\normalfont}
  {}
  {\bfseries}
  {.}
  { }
  {\thmname{#1}\thmnumber{ #2}\thmnote{ (\bfseries #3)}}
\theoremstyle{plain-boldhead}
\newtheorem{theorem}{Theorem}

\newtheorem{lemma}[theorem]{Lemma}

\theoremstyle{definition-boldhead}
\newtheorem{definition}{Definition}

\newtheorem{example}{Example}

\newdateformat{simple}{\THEDAY\ \monthname[\THEMONTH]\ \THEYEAR}
\simple

\floatstyle{ruled}
\newfloat{algo}{htbp}{algo}
\floatname{algo}{Algorithm}

\def \ifempty#1{\def\temp{#1} \ifx\temp\empty }

\newcommand{\str}[1]{\textsc{#1}}
\newcommand{\var}[1]{\textit{#1}}
\newcommand{\op}[1]{\textsl{#1}}

\newcommand{\false}{\textsc{false}\xspace}
\newcommand{\true}{\textsc{true}\xspace}
\newcommand{\etal}{\emph{et al.}}

\newcommand{\BN}{\ensuremath{\mathbb{N}}\xspace}

\newcommand{\CM}{\ensuremath{\mathcal{M}}\xspace}

\newcommand{\CP}{\ensuremath{\mathcal{P}}\xspace}

\newcommand{\CV}{\ensuremath{\mathcal{V}}\xspace}

\newcommand\bcch{\var{bcch}\xspace}

\newcommand\msgs{\var{msgs}\xspace}
\newcommand\inround{\var{inround}\xspace}
\newcommand\vc{\var{vc}\xspace}
\newcommand\cut{\var{cut}\xspace}

\newcommand\emptyhashmap{[\,]}

\newcommand\resend{\var{resend}\xspace}

\newcommand\ppi{\var{$p_i$}\xspace}
\newcommand\ppj{\var{$p_j$}\xspace}

\newcommand\delivered{\var{delivered}\xspace}
\newcommand\ind{\var{index}\xspace}

\newcommand\bargamma{\ensuremath{\overline{\gamma}}\xspace}

\begin{document}

\title{\bf Quick Order Fairness}

\author{Christian Cachin$^1$\\
  University of Bern\\
  \url{cachin@inf.unibe.ch}
  \and Jovana Mićić$^1$\\
University of Bern\\
  \url{jovana.micic@inf.unibe.ch}
  \and Nathalie Steinhauer$^1$\\
University of Bern\\
  \url{nathalie.steinhauer@inf.unibe.ch}
  \and Luca Zanolini$^1$\\
University of Bern\\
  \url{luca.zanolini@inf.unibe.ch}
}

\date{}

\footnotetext[1]{Institute of Computer Science, University of Bern,
  Neubr\"{u}ckstrasse 10, 3012 CH-Bern, Switzerland.}

\maketitle

\begin{abstract}
  \noindent
    Leader-based protocols for consensus, i.e., atomic broadcast, allow some processes to unilaterally
    affect the final order of transactions. This has become a problem for
    blockchain networks and decentralized finance because it facilitates
    front-running and other attacks. To address this, \emph{order fairness}
    for payload messages has been introduced recently as a new safety property
    for atomic broadcast complementing traditional
    \emph{agreement} and \emph{liveness}.
    We relate order fairness to the standard validity notions for consensus
    protocols and highlight some limitations with the existing formalization. 
    Based on this, we introduce a new \emph{differential} order fairness
    property that fixes these issues.
    We also present the \emph{quick order-fair atomic broadcast protocol}
    that guarantees payload message delivery in a differentially fair order and
    is much more efficient than existing order-fair consensus protocols.  It
    works for asynchronous and for eventually synchronous networks with
    optimal resilience, tolerating corruptions of up to one third of the
    processes.  Previous solutions required there to be less than one fourth of
    faults.  Furthermore, our protocol incurs only quadratic cost, in terms
    of amortized message complexity per delivered payload.

    \medskip
    
    \noindent \textbf{Keywords.}
    Consensus, atomic broadcast, decentralized finance, front-running attacks, differential order fairness.
\end{abstract}
\section{Introduction}
\label{sec:intro}

The nascent field of \emph{decentralized finance} (or simply \emph{DeFi})
suffers from insider attacks: Malicious miners in permissionless blockchain
networks or Byzantine leaders in permissioned atomic broadcast protocols
have the power of selecting messages that go into the ledger and
determining their final order.  Selfish participants may also insert their
own, fraudulent transactions and thereby extract value from the network and its innocent users.  For instance, a decentralized exchange can be
exploited by \emph{front-running}, where a genuine message~$m$ carrying an
exchange transaction is \emph{sandwiched} between a message
$m_{\text{before}}$ and a message $m_\text{after}$.  If $m$ buys a
particular asset, the insider acquires it as well using $m_{\text{before}}$
and sells it again with $m_\text{after}$, typically at a higher price.
Such front-running and other price-manipulation attacks represent a serious
threat.  They are prohibited in traditional finance systems with
centralized oversight but must be prevented technically in DeFi. Daian
\etal~\cite{DBLP:conf/sp/DaianGKLZBBJ20} have coined the term \emph{miner
  extractable value (MEV)} for the profit that can be gained from such
arbitrage opportunities.

The traditional properties of \emph{atomic broadcast}, often
somewhat imprecisely called \emph{consensus} as well, guarantee a total order: that all correct parties
obtain the same sequence of messages and that any message submitted to the
network by a client is delivered in a reasonable lapse of time.  However,
these properties do not further constrain \emph{which} order is chosen, and
malicious parties in the protocol may therefore manipulate the order or
insert their own messages to their benefit. Kelkar
\etal~\cite{DBLP:conf/crypto/Kelkar0GJ20} have recently introduced the new
safety property of \emph{order fairness} that addresses this in the Byzantine
model. Kursawe~\cite{DBLP:conf/aft/Kursawe20} and Zhang
\etal~\cite{DBLP:conf/osdi/ZhangSCZA20} have formalized this problem as
well and found different ways to tackle it, relying on somewhat stronger
assumptions.

Intuitively, \emph{order fairness} aims at ensuring that messages received
by ``many'' parties are scheduled and delivered earlier than messages
received by ``few'' parties.  The \emph{Condorcet paradox} demonstrates,
however, that such preference votes can easily lead to cycles even if the
individual votes of majorities are not circular.  The solution offered
through \emph{order fairness}~\cite{DBLP:conf/crypto/Kelkar0GJ20} may
therefore output multiple messages \emph{together as a set} (or batch),
such that there is \emph{no order} among all messages in the same set.
Kelkar \etal~\cite{DBLP:conf/crypto/Kelkar0GJ20} name this property
\emph{block-order fairness} but calling such a set a ``block'' may easily
lead to confusion with the low-level blocks in mining-based
protocols.

In this paper, we investigate order fairness in networks with $n$ processes
of which $f$ are faulty, for asynchronous and eventually synchronous
atomic broadcast.  This covers the vast majority of relevant applications, since
timed protocols that assume synchronous clocks and permanently bounded
message delays have largely been abandoned in this space.

We first revisit the notion of block-order
fairness~\cite{DBLP:conf/crypto/Kelkar0GJ20}.  In our interpretation, this requires that when
$n$ correct processes broadcast two payload messages $m$ and $m'$,
and $\bargamma n$ of them broadcast $m$ before $m'$ for
some $\bargamma > \frac{1}{2}$, then $m'$ is not delivered by the protocol
before~$m$, although both messages may be output together.  This guarantee
is difficult to achieve in practice because Kelkar
\etal~\cite{DBLP:conf/crypto/Kelkar0GJ20} show that for the relevant values
of \bargamma approaching one half, the resilience of any protocol decreases.
Tolerating only a small number of faulty parties seems prohibitive in
realistic settings.

More importantly, we show that $\bargamma$ cannot be too close
to~$\frac{1}{2}$ because $\bargamma \geq \frac{1}{2} + \frac{f}{n-f}$ is
necessary for any protocol.
This result follows from establishing a link to the differential validity
notion of consensus, formalized by Fitzi and
Garay~\cite{DBLP:conf/podc/FitziG03}.  Notice that block-order fairness is
a relative measure.  We are convinced that a differential notion is better
suited to address the problem.  We, therefore, overcome this inherent
limitation of relative order fairness by introducing \emph{differential
  order fairness}: When the number of correct processes that broadcast a
message~$m$ before a message~$m'$ exceeds the number that broadcast $m'$
before $m$ by more than $2f + \kappa$, for some $\kappa \geq 0$, then the
protocol must \emph{not} deliver $m'$ before~$m$ (but they may be delivered
together).  This notion takes into account existing results on differential
validity for consensus~\cite{DBLP:conf/podc/FitziG03}.  In particular, when
the difference between how many processes prefer one of $m$ and $m'$ over
the other is smaller than $2f$, then \emph{no protocol exists} to deliver
them in fair order.

Last but not least, we introduce a new protocol, called \emph{quick
  order-fair atomic broadcast}, that implements differential order fairness
and is much more efficient than the previously existing algorithms.  In
particular, it works with optimal resilience $n > 3f$, requires $O(n^2)$
messages to deliver one payload on average and needs
$O(n^2 L + n^3 \lambda)$ bits of communication, with payloads of up to $L$
bits and cryptographic $\lambda$-bit signatures. This holds for \emph{any}
order-fairness parameter~$\kappa$.  For comparison, the asynchronous
Aequitas protocol~\cite{DBLP:conf/crypto/Kelkar0GJ20} has resilience $n > 4f$ or worse, depending on its order-fairness parameter, and needs $O(n^4)$ messages.

To summarize, the contributions of this paper are as follows:
\begin{itemize}
\item It illustrates some \emph{limitations} that are inherent in the
  notion of block-order fairness (Section~\ref{sec:limitations}).
\item It introduces \emph{differential order fairness} as a measure for
  defining fair order in atomic broadcast protocols
  (Section~\ref{sec:differential}).
\item It presents the \emph{quick order-fair atomic broadcast protocol} for
differentially order-fair Byzantine atomic broadcast with optimal resilience
  $n > 3f$ (Section~\ref{sec:protocol}).
\item It demonstrates that the quick order-fairness protocol has quadratic
  amortized message complexity, which is an $n^2$-fold improvement compared
  to the most efficient previous protocol for the same task
  (Section~\ref{subsection:complexity}).
\end{itemize}
The paper starts with a review of previous work (Section~\ref{sec:related})
and by describing our system model (Section~\ref{sec:model}).

\section{Related work}
\label{sec:related}

Over the last decades, extensive research efforts have explored the state-machine replication problem. A large number of papers refer to this problem, but only a few of them consider fairness in the order of delivered payload messages. In this section, we review the related work on fairness.

Kelkar \etal~\cite{DBLP:conf/crypto/Kelkar0GJ20} introduce a new property called \emph{transaction order-fairness} which prevents adversarial manipulation of the ordering of transactions, i.e., payload messages. They investigate assumptions needed for achieving this property in a permissioned setting and formulate a new class of consensus protocols, called Aequitas, that satisfy order fairness. A subsequent paper by Kelkar \etal~\cite{DBLP:journals/iacr/KelkarDK21} extends this approach to a permissionless setting. Recently, Kelkar \etal~\cite{DBLP:journals/iacr/KelkarDLJK21} presented another permissioned Byzantine atomic-broadcast protocol called Themis. It introduces a new technique called \emph{deferred ordering}, which overcomes a liveness problem of the Aequitas protocols.

Kursawe~\cite{DBLP:conf/aft/Kursawe20} and Zhang \etal~\cite{DBLP:conf/osdi/ZhangSCZA20} have independently postulated alternative definitions of order fairness, called \emph{timed order fairness} and \emph{ordering linearizability}, respectively. Both notions are strictly weaker than order fairness of transactions, however~\cite{DBLP:journals/iacr/KelkarDK21}. Timed order fairness assumes that all processes have access to synchronized local clocks; it can ensure that if all correct processes saw message $m$ to be ordered before $m'$, then $m$ is scheduled and delivered before $m'$. Similarly, ordering linearizability says that if the highest timestamp provided by any correct process for a message $m$ is lower than the lowest timestamp provided by any correct process for a message $m'$, then $m$ will appear before $m'$ in the output sequence. The implementation of ordering linearizability~\cite{DBLP:conf/osdi/ZhangSCZA20} uses a median computation, which can easily be manipulated by faulty processes~\cite{DBLP:journals/iacr/KelkarDK21}. 

The Hashgraph~\cite{baird2016swirlds} consensus protocol also claims to achieve fairness. It uses gossip internally and all processes build a \emph{hash graph} reflecting all of the gossip events. However, there is no formal definition of fairness and the presentation fails to recognize the impossibility of fair message-order resulting from the \emph{Condorcet paradox}. Kelkar \etal~\cite{DBLP:conf/crypto/Kelkar0GJ20} also show an attack that allows a malicious process to control the order of the messages delivered by Hashgraph.

A complementary measure to prevent message-reordering attacks relies on
threshold cryptography~\cite{DBLP:journals/toplas/ReiterB94,DBLP:conf/crypto/CachinKPS01,DBLP:conf/dsn/DuanRZ17}:
clients encrypt their input (payload) messages under a key shared by the group of
processes running the atomic broadcast protocol.  They initially order the encrypted
messages and subsequently collaborate for decrypting them. Hence, their
contents become known only \emph{after} the message order has been decided.
For instance, the Helix protocol~\cite{DBLP:conf/icnp/AsayagCGLRTY18}
implements this approach and additionally exploits in-protocol randomness
for two additional goals: to elect the processes running the protocol from a
larger group and to determine which messages among all available ones
must be included by a process when proposing a block.  This method provides
resistance to censorship but still permits some order-manipulation
attacks.

\section{System model and preliminaries}
\label{sec:model}

\subsection{System model}

\paragraph{Processes.} We model our system as a set of $n$ \emph{processes} $\CP = \{p_1, \dots, p_n\}$, also called \emph{parties}, that communicate with each other. Processes interact with each other by exchanging messages reliably in a network. A protocol for \CP consists of a collection of programs with instructions for all processes. Processes are computationally bounded and protocols may use cryptographic primitives, in particular, digital signature schemes.

\paragraph{Failures.} In our model, we distinguish two types of processes. Processes that follow the protocol as expected are called \emph{correct}. Contrary, the processes that deviate from the protocol specification or may crash are called \emph{Byzantine}.

\paragraph{Communication.} We assume that there exists a low-level mechanism for sending messages over reliable and authenticated point-to-point links between processes. In our protocol implementation, we describe this as ``sending a message" and ``receiving a message". Additionally, we assume \emph{first-in first-out (FIFO) ordering} for the links. This ensures that messages broadcast by the same correct process are delivered in the order in which they were sent by a correct recipient.

\paragraph{Timing.} This work considers two models, \emph{asynchrony} and \emph{partial synchrony}. Together they cover most scenarios used today in the context of secure distributed computing. In an \emph{asynchronous} network, no physical clock is available to any process and the delivery of messages may be delayed arbitrarily. In such networks, it is only guaranteed that a message sent by a correct process will \emph{eventually} arrive at its destination. One can define asynchronous time based on logical clocks. A \emph{partially synchronous} network~\cite{DBLP:journals/jacm/DworkLS88} operates asynchronously until some point in time (not known to the processes), after which it becomes stable. This means that processing times and message delays are bounded afterwards, but the maximal delays are not known to the protocol.

\subsection{Byzantine FIFO Consistent Broadcast Channel}

We are using a Byzantine FIFO consistent broadcast channel (BCCH) that allows the processes to deliver multiple payloads and ensures a notion of consistency despite Byzantine senders. The interface of such a channel provides two events involving payloads from a domain~\CM:
\begin{itemize}
  \item A process invokes $\op{\bcch-broadcast}(m)$ to broadcast a message $m \in \CM$ to all processes.
  \item An event $\op{\bcch-deliver}(\ppj, l, m)$ delivers a message $m \in \CM$ with label~$l \in \left\{0,1\right\}^* $ from a process~\ppj.
\end{itemize}
The label that comes with every delivered message is an arbitrary bit string generated by the channel. Intuitively, the channel ensures that if a message is delivered with some label, then the message itself is the same at all correct processes that deliver this label.

\begin{definition}[Byzantine FIFO Consistent Broadcast Channel]\label{def:bcch}
  A Byzantine FIFO consistent broadcast channel satisfies the following properties:
 \begin{description}
 \item[Validity:]  If a correct process broadcasts a message $m$, then every correct process eventually delivers $m$.
 \item[No duplication:] For every process $p_j$ and label $l$, every correct process delivers at most one message with label $l$ and sender $p_j$.
 \item[Integrity:] If some correct process delivers a message $m$ with sender $p_j$ and process $p_j$ is correct, then $m$ was previously broadcast by $p_j$.
 \item[Consistency:] If some correct process delivers a message $m$ with label $l$ and sender \ppj, and another correct process delivers a message $m'$ with label $l$ and sender \ppj, then $m = m'$.
 \item[FIFO delivery:] If a correct process broadcasts some message~$m$ before it broadcasts a message~$m'$, then no correct process delivers $m'$ unless it has already delivered $m$.
 \end{description}
\end{definition}

This primitive can be implemented by running, for every sender~$p_i$, a sequence of standard consistent Byzantine broadcast instances~\cite[Sec.~3.10]{DBLP:books/daglib/0025983} such that exactly one instance in each sequence is active at every moment. Each consistent broadcast instance is identified by a per-sender sequence number. When an instance delivers a message, the protocol advances the sequence number and initializes the next instance.  The sequence number serves as the label. Details of this protocol are described by Cachin \etal~\cite[Sec.~3.12.2]{DBLP:books/daglib/0025983}; notice that their protocol also ensures FIFO delivery, although this is not explicitly mentioned there.

In addition to the $\op{bcch-broadcast}$ and $\op{bcch-deliver}$ events, in our protocol we use the following methods to access the BCCH primitive: $\op{bcch-create-proof}$ and $\op{bcch-verify-proof}$. Those methods ensure that missing messages can be transferred in a verifiable way, and they are implemented as in the protocol for verifiable consistent broadcast by Cachin \etal~\cite{DBLP:conf/crypto/CachinKPS01}.
The input of \op{bcch-create-proof} is a list of messages and it outputs a string $s$ that contains a proof along with the list of messages to be sent. A process that receives a message providing $s$ can input this in \op{bcch-verify-proof} to verify the proof contained in $s$ such that it is impossible to forge a proof for a message that was not \op{bcch-delivered}.

Another two methods, $\op{bcch-get-length}$ and $\op{bcch-get-messages}$, are used to get the number of sent payload messages and to extract them.

\subsection{Validated Byzantine Consensus}
\label{sec:vbc}

Validated Byzantine consensus~\cite{DBLP:conf/crypto/CachinKPS01} defines an \textit{external validity} condition. It requires that the consensus value is legal according to a global, efficiently computable predicate~$P$, known to all processes. This allows the protocol to recognize proposed values that are acceptable to an external application. Note that it is not required that the decision value was proposed by a correct process, but all processes must be able to verify the validity. A consensus primitive is accessed through the events $\op{vbc-propose}(v)$ and $\op{vbc-decide}(v)$, where $v \in \CV$ has a potentially large domain~\CV and may contain a proof, which allows processes to verify the validity of $v$.

\begin{definition}[Validated Byzantine Consensus]\label{def:vbc}
 A protocol solves validated Byzantine consensus with validity predicate~$P$ if it satisfies the following conditions:
 \begin{description}
 \item[Termination:] Every correct process eventually decides some value.
 \item[Integrity:] No correct process decides twice.
 \item[Agreement:] No two correct processes decide differently.
 \item[External validity:] Every correct process only decides a value $v$ such that $P(v) = \true$. Moreover, if all processes are correct and propose $v$, then no correct process decides a value different from $v$.
\end{description}
\end{definition}

We intend this notion to cover asynchronous protocols, which actually only terminate probabilistically, as well as eventually synchronous protocols. The difference is not essential to our use of them.

Originally, \emph{external validity} has been defined for \emph{asynchronous} multi-valued Byzantine consensus, which requires randomized implementations~\cite{DBLP:conf/crypto/CachinKPS01}.  But the property applies equally to consensus protocols with \emph{partial synchrony}.

Among the asynchronous protocols, recent work by Abraham \etal~\cite{DBLP:conf/podc/AbrahamMS19} improves the expected communication (bit) complexity to $O(L n^2)$ from $O(L n^3)$ in the earlier work~\cite{DBLP:conf/crypto/CachinKPS01}, where $L$ is the maximal length of a proposed value.

In Dumbo-MVBA~\cite{DBLP:conf/podc/LuL0W20} the communication complexity  of this primitive is further reduced to $O(Ln)$ through erasure coding, where the input of each process is split into coded fragments, distributed to every process, and recovered later.

Byzantine consensus protocols in the partial-synchrony model can easily be enhanced to provide external validity, when each process verifies $P$ for every proposed value. For instance, the single-decision versions of PBFT~\cite{DBLP:journals/tocs/CastroL02} and of HotStuff~\cite{DBLP:conf/podc/YinMRGA19} achieve best-case complexities $O(L n^2)$ and $O(L n)$, respectively; these values increase by a factor of~$n$ in the worst case.

\subsection{Atomic Broadcast}

Atomic broadcast ensures that all processes deliver the same messages and that all messages are output in the same order.  This is equivalent to the processes agreeing on one sequence of messages that they deliver.  Atomic broadcast is also called ``total-order broadcast'' or simply ``consensus'' in the context of blockchains because it is equivalent to running a sequence of consensus instances.  Processes may broadcast a message~$m$ by invoking $\op{a-broadcast}(m)$, and the protocol outputs messages through $\op{a-deliver}(m)$ events.
\begin{definition}[Atomic Broadcast]\label{def:abc}
  A protocol for atomic broadcast satisfies the following properties:
 \begin{description}
  \item[Validity:]  If a correct process \op{a-broadcasts} a message $m$, then every correct process eventually \op{a-delivers}~$m$.
  \item[No duplication:] No message is \op{a-delivered} more than once.
  \item[Agreement:] If a message $m$ is \op{a-delivered} by some correct process, then $m$ is eventually \op{a-delivered} by every correct process.
  \item[Total order:] Let $m$ and $m'$ be two messages such that \ppi and \ppj are correct processes that \op{a-deliver} $m$ and $m'$. If \ppi \op{a-delivers} $m$ before $m'$, then \ppj also \op{a-delivers} $m$ before $m'$.
 \end{description}
\end{definition}

\section{Revisiting order fairness}
\label{sec:orderfairness}
  
In this section, we discuss the challenges of defining order fairness and highlight limitations of order fairness notions from previous works. We then introduce our refined notion of differential order-fair atomic broadcast. 

\subsection{Limitations}
\label{sec:limitations}

Defining a fair order for atomic broadcast in asynchronous networks is not straightforward since the processes might locally receive messages for broadcasting in different orders. We assume here that a correct process receives a payload to be broadcast (e.g., from a client) at the same time when it \op{a-broadcasts} it.
If a process broadcasts a payload message $m$ before a payload message $m'$, according to its local order, we denote this by $m \prec m'$.
Furthermore, we abandon the \emph{validity} property above in the context of atomic broadcast with order fairness and assume now that every payload message is a-broadcast by all correct processes. This corresponds to the implicit assumption made for deploying order-fair broadcast.

Even if all processes are correct, it can be impossible to define a fair order among all messages. This is shown by a result from social science, known as the \emph{Condorcet paradox}, which states that there exist situations that lead to non-transitive collective voting preferences even if the individual preferences are transitive.  Kelkar \etal~\cite{DBLP:conf/crypto/Kelkar0GJ20} apply this to atomic broadcast and show that delivering messages in a fair order is not always possible.  Their example considers three correct processes $p_1$, $p_2$, and $p_3$ that receive three payload messages $m_a$, $m_b$, and $m_c$. While $p_1$ receives these payload messages in the order \(m_a \prec m_b \prec m_c\), process~$p_2$ receives them as \(m_b \prec m_c \prec m_a\) and $p_3$ in the order \(m_c \prec m_a \prec m_b\). Obviously, a majority of the processes received $m_a$ before $m_b$, $m_b$ before $m_c$, but also $m_c$ before $m_a$, leading to a cyclic order. Consequently, a fair order cannot be specified even with only correct processes.

One way to handle situations with such cycles in the order is presented by Kelkar \etal~\cite{DBLP:conf/crypto/Kelkar0GJ20} with \emph{block-order fairness}: their protocol delivers a ``block'' of payload messages at once. Typically, a block will contain those payloads that are involved in a cyclic order. Their notion requires that if sufficiently many processes receive a payload $m$ before another payload $m'$, then no correct process delivers $m$ after $m'$, but they may both appear in the same block.
Even though the order among the messages within a block remains unspecified, the notion of block-order fairness respects a fair order up to this limit.

Kelkar \etal~\cite{DBLP:conf/crypto/Kelkar0GJ20} specify ``sufficiently many'' as a $\gamma$-fraction of \emph{all} processes, where $\gamma$ represents an order-fairness parameter such that $\frac{1}{2} < \gamma \leq 1$. More precisely, block-order fairness considers a number of processes $\eta$ that all receive (and broadcast) two payload messages $m$ and $m'$. Block-order fairness for atomic broadcast requires that whenever there are at least $\gamma \eta$ processes that receive $m$ before $m'$, then no correct process delivers $m$ after~$m'$ (but they may deliver $m$ and $m'$ in the same block).

Kelkar \etal~\cite{DBLP:conf/crypto/Kelkar0GJ20} explicitly count faulty processes for their definition.  Notice that this immediately leads to problems: If $\gamma \eta < 2f$, for instance, the notion relies on a majority of faulty processes, but no guarantees are possible in this case.  Therefore, we only count on events occurring at correct processes here and define a block-order fairness parameter \bargamma to denote the fraction of \emph{correct} processes that receive one message before the other.

Moreover, we assume w.l.o.g.\ that all correct processes eventually broadcast every payload, even if this is initially input by a single process only.  This simplifies the treatment compared to original block-order fairness, which considers only processes that broadcast \emph{both} payload messages, $m$ and $m'$~\cite{DBLP:conf/crypto/Kelkar0GJ20}.  Our simplification means that a correct process that has received only one payload will receive the other payload as well later.  This process should eventually include also the second payload for establishing a fair order.  It corresponds to how atomic broadcast is used in practice; hence, we set $\eta = n-f$. In asynchronous networks, furthermore, one has to respect $f$ additional correct processes that may be delayed. Their absence reduces the strength of the formal notion of block-order fairness in asynchronous networks even more.

In the following, we discuss the range of achievable values for $\bargamma$. Since we focus on models that allow asynchrony, we assume $n > 3f$ throughout this work.  Fundamental results on validity notions for Byzantine consensus in asynchronous networks have been obtained by Fitzi and Garay~\cite{DBLP:conf/podc/FitziG03}.  Recall that a consensus protocol satisfies \emph{termination}, \emph{integrity}, and \emph{agreement} according to Definition~\ref{def:vbc}. \emph{Standard consensus} additionally satisfies:
\begin{description}
\item[Validity:] If all correct processes propose~$v$, then all correct processes decide~$v$.
\end{description}
Notice that this leaves the decision value completely open if only one correct process proposes something different.  In their notion of \emph{strong consensus}, however, the values proposed by correct processes must be better respected, under more circumstances:
\begin{description}
\item[Strong validity:] If a correct process decides~$v$, then some correct process has proposed~$v$.
\end{description}
Unfortunately, strong consensus is not suitable for practical purposes because Fitzi and Garay~\cite[Thm.~8]{DBLP:conf/podc/FitziG03} also show that if the proposal values are taken from a domain~\CV, then the resilience depends on $|\CV|$.  In particular, strong consensus is only possible if $n > |\CV| f$.

Related to this, they also introduce \emph{$\delta$-differential consensus}, which respects how many times a value is proposed by the correct processes. This notion ensures, in short, that the decision value has been proposed by ``sufficiently many'' correct processes compared to how many processes proposed some different value.  More precisely, for an execution of consensus and any value~$v \in \CV$, let $c(v)$ denote the number of correct processes that propose~$v$:
\begin{description}
\item[$\delta$-differential validity:] If a correct process decides $v$, then every other value~$w$ proposed by some correct process satisfies $c(w) \leq c(v) + \delta$.
\end{description}
To summarize, whereas the standard notion of Byzantine consensus requires that \emph{all} correct processes start with the same value in order to decide on one of the correct processes' input, strong consensus achieves this in any case. It requires that the decision value has been proposed by \emph{some} correct process. However, it does not connect the decision value to how many correct processes have proposed it. Consequently, strong consensus may decide a value proposed by just one correct process. Differential consensus, finally, makes the initial plurality of the decision value explicit. For $\delta = 0$, in particular, the decision value must be one of the proposed values that is most common among the correct processes. More importantly, differential validity can be achieved under the usual assumption that $n>3f$.

We now give another characterization of $\delta$-differential validity.
For a particular execution of some (asynchronous) Byzantine consensus
protocol, let $v^*$ be (one of) the value(s) proposed most often by correct
processes, i.e.,
\[
  v^* = \arg \max_v c(v).
\]

\begin{lemma}\label{lem:differential}
  A Byzantine consensus protocol satisfies $\delta$-differential validity
  if and only if in every one of its executions, it never decides a value
  $w$ with $c(w) < c(v^*) - \delta$.
\end{lemma}

\begin{proof}
  Assume first that the protocol satisfies $\delta$-differential validity
  and a correct process decides any value~$v$ in the domain.  Then every
  other value~$w$ proposed by a correct process satisfies
  $c(w) \leq c(v) + \delta$.  In particular, this implies
  $c(v^*) \leq c(v) + \delta$, which is equivalent to,
  $c(v) \geq c(v^*) - \delta$.  Hence, the protocol \emph{never} decides a
  value $x$ with $c(x) < c(v^*) - \delta$.
  
  To show the reverse direction, suppose $x$ is such that
  $c(x) < c(v^*) - \delta$ and a correct process decides~$x$.  This does
  not satisfy $\delta$-differential validity because also $v^*$ has been
  proposed by a correct process but $c(v^*) > c(x) + \delta$.
\end{proof}

For consensus with a \emph{binary} domain $\CV = \{0, 1\}$, this means that
a consensus protocol satisfies $\delta$-differential validity if and only
if in every one of its executions with, say, $c(0) > c(1) + \delta$, every
correct process decides~$0$.

No asynchronous consensus algorithm for agreeing on the value proposed by a simple majority of correct processes exists, however.  Fitzi and Garay \cite[Thm.~11]{DBLP:conf/podc/FitziG03} prove that $\delta$-differential consensus in asynchronous networks is \emph{not possible} for $\delta < 2f$:
\begin{theorem}[\cite{DBLP:conf/podc/FitziG03}]\label{thm:fitzi}
  In an asynchronous network, $\delta$-differential consensus is achievable only if $\delta \geq 2f$.
\end{theorem}

The above discussion already hints at issues with achieving fair order in asynchronous systems.  Recall that Kelkar \etal~\cite{DBLP:conf/crypto/Kelkar0GJ20} present atomic broadcast protocols with block-order fairness for the asynchronous setting with order-fairness parameter~$\gamma$ (whose definition includes faulty processes). The corruption bound is stated as
\begin{equation}\label{equ:corruption}
  n > \frac{4f}{2\gamma-1}.
\end{equation}
For $\gamma = 1$, which ensures fairness only in the most clear cases, there are $n > 4f$ processes required.  For values of $\gamma$ close to~$\frac{1}{2}$, the condition becomes prohibitive for practical solutions.

In fact, even when using our interpretation, $\bargamma$ cannot be too close to~$\frac{1}{2}$, as the following result shows.  It rules out the existence of $\bargamma$-block-order-fair atomic broadcast in asynchronous or eventually synchronous networks for $\bargamma < \frac{1}{2} + \frac{f}{n-f}$.

\begin{theorem}\label{thm:gamma}
  In an asynchronous network with $n$ processes and $f$ faults, implementing atomic broadcast with $\bargamma$-fair block order is not possible unless $\bargamma \geq \frac{1}{2} + \frac{f}{n-f}$.
\end{theorem}
\begin{proof}
  Towards a contradiction, suppose there is an atomic broadcast protocol
  ensuring $\bargamma$-fair block order with
  $\frac{1}{2} < \bargamma < \frac{1}{2} + \frac{f}{n-f}$.  We will transform
  this into a differential consensus protocol that violates
  Theorem~\ref{thm:fitzi}.

  The consensus protocol works like this.  All processes initialize the
  atomic broadcast protocol.  Upon $\op{propose}(v)$ with some value~$v$, a
  process simply \op{a-broadcasts} $v$.  When the first value $v'$ is
  \op{a-delivered} by atomic broadcast to a process, the process executes
  $\op{decide}(v')$ and terminates.
  
  Consider any execution of this protocol such that all correct processes
  propose one of two values, $m$ or $m'$.  Suppose w.l.o.g.\ that
  $c(m) = \bargamma (n-f)$ and $c(m') = (1-\bargamma) (n-f)$, i.e., $m$ is
  proposed $c(m)$ times by correct processes and more often than $m'$,
  since $\bargamma > \frac{1}{2}$.  It follows that $\bargamma (n-f)$ correct
  processes \op{a-broadcast} $m$ before $m'$ and $(1-\bargamma) (n-f)$ correct
  processes \op{a-broadcast}~$m'$ before~$m$.

  According to the properties of atomic broadcast all correct processes
  \op{a-deliver} the same value first in every execution.  Moreover, the
  atomic broadcast protocol \op{a-delivers} $m$ before $m'$ by the
  $\bargamma$-fair block order property.  This implies that the consensus
  protocol decides~$m$ in every execution and never~$m'$.  Since no further
  restrictions are placed on $m$ and on $m'$, this consensus protocol
  actually ensures $\delta$-differential validity for some
  $\delta < c(m) - c(m')$ by Lemma~\ref{lem:differential}.
  
  However, the $c(m)$ and $c(m')$ satisfy, respectively,
  \[
    \begin{array}{ccccccl}
      c(m) & = &\bargamma(n-f)
      & < &\Bigl(\frac{1}{2} + \frac{f}{n-f}\Bigr) (n-f)
      & = &\frac{n+f}{2} \\
      c(m') & = &(1-\bargamma)(n-f)
      & > &\Bigl(1 - \frac{1}{2} - \frac{f}{n-f}\Bigr) (n-f)
      & = &\frac{n-3f}{2}
    \end{array}
  \]
  and, therefore,
  $\delta < c(m) - c(m') < \frac{n+f}{2} - \frac{n-3f}{2} = 2f$.  But
  $\delta$-differential asynchronous consensus is only possible when
  $\delta \geq 2f$, a contradiction.
\end{proof}

\subsection{Differential Order-Fairness}
\label{sec:differential}

The limitations discussed above have an influence on order fairness. The condition on $\delta$ to achieve \emph{$\delta$-differential consensus} directly impacts any measure of fairness.
It becomes clear that a \emph{relative} notion for block-order fairness, defined through a fraction like~$\bargamma$, may not be expressive enough.

We now start to define our notion of \emph{order-fair atomic broadcast}; it has almost the same interface as regular atomic broadcast. The primitive is accessed with $\op{of-broadcast}(m)$ for broadcasting a payload message~$m$ and it outputs payload messages through $\op{of-deliver}(M)$ events, where $M$ is a \emph{set} of payloads delivered at the same time; $M$ corresponds the block of block-order fairness. We want to count the number of correct processes that \op{of-broadcast} a message $m$ before another message $m'$ and introduce a function
\[
  b : \CM \times \CM \to \BN
\]
for all $m$ and $m'$ that were ever \op{of-broadcast} by correct processes. The value $b(m,m')$ denotes the \emph{number of correct processes} that \op{of-broadcast} $m$ before $m'$ in an execution.  As above we assume w.l.o.g.\ that a correct process will \op{of-broadcast} $m$ and $m'$ eventually and that, therefore, $b(m, m') + b(m',m) = n-f$.

Can we achieve that if $b(m,m') > b(m',m)$, i.e., when there are more correct processes that \op{of-broadcast}  message $m$ before $m'$ than correct processes that \op{of-broadcast} $m'$ before $m$, then no correct process will \op{of-deliver} $m'$ before $m$? Using a reduction from $\delta$-differential consensus, as in the previous result, we can show that this condition is too weak.

\begin{theorem}\label{thm:mu}
  Consider an atomic broadcast protocol that satisfies the following notion
  of order fairness for some $\mu \geq 0$:
  \begin{description}
  \item[Weak differential order fairness:] For any $m$ and $m'$, if
    $b(m,m') > b(m',m) + \mu$, then no correct process \op{a-delivers} $m'$
    before~$m$.
  \end{description}
  Then it must hold $\mu \geq 2f$.
\end{theorem}

\begin{proof}
  Towards a contradiction, suppose there is an atomic broadcast protocol,
  which ensures that for all payload messages $m$ and $m'$ with
  $b(m,m') > b(m',m) + \mu$ and $\mu \geq 0$, no correct process
  \op{a-delivers} $m'$ before $m$ and that $\mu < 2f$.
  We will transform this into a differential consensus protocol that
  violates Theorem~\ref{thm:fitzi}.
 
  The consensus protocol works like this. All processes initialize the
  order-fair atomic broadcast protocol. Upon \op{propose}$(v)$ with some
  value $v$, a process simply \op{of-broadcasts} $v$. When the first value
  $v'$ is \op{of-delivered} to a process, the process executes
  \op{decide}$(v')$ and terminates.
 
  Consider any execution of this protocol such that all correct processes
  propose one of two values, $m$ or $m'$.  Suppose w.l.o.g.\ that $m$ is
  proposed $c(m)$ times by correct processes and more often than $m'$,
  which is proposed $c(m')$ times, with $c(m) + c(m') = n - f$ and that
  $c(m) > c(m') + \mu$.  It follows that $b(m,m') = c(m)$ correct processes
  \op{of-broadcast} $m$ before $m'$ and $b(m',m) = c(m')$ correct processes
  \op{of-broadcast} $m'$ before $m$, hence, $b(m,m') > b(m',m) + \mu$.
 
  According to the properties of atomic broadcast, all processes
  \op{of-deliver} the same value first in every execution. Moreover, the
  protocol \op{of-delivers} $m$ before $m'$ because
  $b(m,m') > b(m',m) + \mu$.  This implies that the consensus protocol
  decides $m$ in every such execution. Since no further restrictions are
  placed on $m$ and $m'$ and since $c(m) - c(m') > \mu$, this consensus
  protocol actually implements $\mu$-differential consensus by
  Lemma~\ref{lem:differential}.  However, achieving $\mu$-differential
  asynchronous consensus requires that $\mu \geq 2f$ according to
  Theorem~\ref{thm:fitzi}.  But $\mu < 2f$ by the above assumption.  This
  is a contradiction.
\end{proof}

On the basis of this result, we now formulate our notion of
\emph{$\kappa$-differentially order-fair atomic broadcast}, using a
fairness parameter $\kappa \geq 0$ to express the strength of the fairness.
Smaller values of $\kappa$ ensure stronger fairness in the sense of how
large the majority of processes that \op{of-broadcast} some $m$ before $m'$
must be to ensure that $m$ will be \op{of-delivered} before $m'$ and in a
fair order.

Recall that throughout this work, we assume that if one correct process
\op{of-broadcasts} some payload~$m$, then every correct process eventually
also \op{of-broadcasts}~$m$.
For reasons that are discussed later (in the remarks after the protocol
description in Section~\ref{subsection:details}), we use a weaker formal
notion of validity, which considers executions with only correct processes.

\begin{definition}[$\kappa$-Differentially Order-Fair Atomic Broadcast] \label{def:kappa-of}
  A protocol for \emph{$\kappa$-differentially order-fair atomic broadcast}
  satisfies the properties \emph{no duplication, agreement} and \emph{total order} of \emph{atomic broadcast} and additionally:
  \begin{description}
  \item [Weak validity:] If all processes are correct and \op{of-broadcast}
    a finite number of messages, then every correct process eventually
    \op{of-delivers} all of these \op{of-broadcast} messages.
  \item [$\kappa$-differential order fairness:] If
    $b(m,m') > b(m',m) + 2f + \kappa$, then no correct process
    \op{of-delivers} $m'$ before $m$.
  \end{description}
\end{definition}
Compared to the above notion of weak differential order fairness, we have
$\kappa = \mu - 2f$.  We show in the next section how to implement
$\kappa$-differentially order-fair atomic broadcast.

\section{Quick order-fair atomic broadcast protocol}
\label{sec:protocol}

This section presents first an overview of our \textit{quick order-fair atomic broadcast} algorithm in Section~\ref{subsection:overview}. A detailed description of the implementation follows in Section~\ref{subsection:details}, along with the pseudocode in Algorithm~\ref{alg:alg1}--\ref{alg:alg2}.  Finally, the complexity of the algorithm is discussed in Section \ref{subsection:complexity}.

\subsection{Overview}
\label{subsection:overview}

The protocol concurrently runs a Byzantine FIFO consistent broadcast channel (BCCH) and proceeds in rounds of consensus.  BCCH allows processes to deliver multiple messages consistently.  An incoming \op{of-broadcast} event with a payload message $m$ triggers BCCH and \op{\bcch-broadcasts}~$m$ to the network.  Additionally, every process keeps a local vector clock that counts the payloads that have been \op{\bcch-delivered} from each sending process.  Every process also maintains an array of lists \msgs such that $\msgs[i]$ records all \op{\bcch-delivered} payloads from~$p_i$.

When a process \op{\bcch-delivers} the payload message $m$, it increments the corresponding vector-clock entry and appends $m$ to the appropriate list in \msgs.
As soon as sufficiently many new payloads are found in \msgs, a new round starts.  Each process signs its vector clock and sends it to all others. The received vector clocks are collected in a matrix, and once $n-f$ valid vector clocks are recorded, a new validated Byzantine consensus (VBC) instance is triggered. The process proposes the matrix and the signatures for consensus, and VBC decides on a common matrix with valid signatures.
This matrix defines a \emph{cut}, which is a vector of indices, with one index per process, such that the index for $p_j$ determines an entry in $\msgs[j]$ up to which payload messages are considered for creating the fair order in the round.  It may be that the index points to messages that a process~\ppi does not store in $\msgs[j]$ because they have not been \op{bcch-delivered} yet.  When the process detects such a missing payload, it asks all other processes to send the missing payload directly and in a verifiable way, such that every process will store all payloads up to the cut in \msgs.

Once all processes received the payloads up to the cut, the algorithm starts to build a graph that represents the dependencies among messages that must be respected for a fair order.  This graph resembles the one used in Aequitas~\cite{DBLP:conf/crypto/Kelkar0GJ20}, but its semantics and implementation differ.  The vertices in the graph here are all \emph{new} payload messages defined by the cut and an edge $(m, m')$ indicates that $m$ should at most be \op{of-delivered} before $m'$.

The graph results from two steps.  In the first step, the process creates a vertex for every payload message that appears in a distinct lists in \msgs and it is not yet \op{of-delivered}.  In the second step, the algorithm builds a matrix $M$ such that $M[m][m']$ counts how many times $m$ appears before $m'$ in \msgs (up to the cut).  $M[m][m']$ can be interpreted as \emph{votes}, counting how many processes want to order $m$ before~$m'$.  Notice that entries of $M$ exist only for $m$ and $m'$ where at least one of $M[m][m']$ and $M[m'][m]$ is non-zero.

If the \emph{difference} between entries $M[m][m']$ and $M[m'][m]$ is large enough, then the protocol adds a directed edge $(m,m')$ to the graph.  The edge indicates that $m'$ must not be \op{of-delivered} before $m$.  More precisely, assuming that messages $m$ and $m'$ have been observed by at least $n-f$ processes, such an edge is added for all $m$ and $m'$ with $M[m][m'] > M[m'][m] - f + \kappa$.  The condition is explained through the following result.

\begin{lemma} \label{lem:M}
  If $b(m,m') > b(m',m) +2f + \kappa$, then $M[m][m'] > M[m'][m] - f + \kappa$.
\end{lemma}

\begin{proof}
  At least $M[m][m']-f$ correct processes have \op{of-broadcast} $m$ before $m'$ because $M[m][m']$ may include reports about $m$ and $m'$ in \var{msgs} from up to $f$ incorrect processes. In other words,
  \begin{align*}
    b(m,m') \geq M[m][m']-f
    \quad\Longleftrightarrow\quad
    M[m][m'] \leq b(m,m')+f
  \end{align*}
  At most $M[m][m']+2f$ correct processes have \op{of-broadcast} $m$ before $m'$ because $M[m][m']$ may include reports about $m$ or $m'$ in \var{msgs} from up to $f$ incorrect processes, and there may be up to $2f$ correct processes whose arrays were not considered in this number. That is,
  \begin{align*}
    b(m,m') \leq M[m][m']+2f
    \quad\Longleftrightarrow\quad
    M[m][m'] \geq b(m,m')-2f
  \end{align*}
  Suppose $b(m,m') > b(m',m) +2f + \kappa$. The above implies
  \begin{align*}
    M[m][m'] &\geq b(m,m')-2f \\
             &> b(m',m)+2f+\kappa-2f \\
             &= b(m',m)+\kappa \\
             &\geq M[m'][m]-f+\kappa.
  \end{align*}
  Thus, whenever $M[m][m'] > M[m'][m]-f+\kappa$, we need to prevent that the protocol \op{of-delivers} $m'$ before $m$. We do this by adding an edge from $m$ to $m'$ to the graph; as shown later, this ensures that $m'$ is not \op{of-delivered} before $m$.
\end{proof}

In the discussion so far, we have assumed that the two messages $m$ and $m'$ were received by at least $n-f$ processes.  Observe that every process can only contribute with 1 to either $M[m][m']$ or to $M[m'][m]$, but not to both.  However, it may occur that only a few processes receive $m$ and $m'$ before the cut, which implies that $M[m][m']$ may be very small, for example.  But that count might actually grow later and take on values up to $n - f - M[m'][m]$.  For this reason, we extend the condition derived from Lemma~\ref{lem:M} in the algorithm as follows: if $n-f-M[m'][m] > M[m'][m]-f+\kappa$ (which implies that $M[m'][m]$ is small, i.e., $M[m'][m] < \frac{n-\kappa}{2}$), we also add add an edge between $m$ and $m'$.  In summary, then, the algorithm adds an edge from $m$ to $m'$ whenever
\[
  \max\big\{M[m][m'], n-f-M[m'][m]\big\} > M[m'][m] -f + \kappa.
\]

Creating the graph in this manner leads to a directed graph that represents constraints to be respected by a fair order. Notice that two messages may be connected by edges in both directions when the difference is small and $\kappa < f$, i.e., there may be a cycle $(m,m')$ and $(m',m)$. This means that the difference between the number of processes voting for one or the other order is too small to decide on a fair order.  Longer cycles may also exist.  All payload messages with circular dependencies among them will be \op{of-delivered} together as a set.  For deriving this information, the algorithm repeatedly detects all strongly connected components in the graph and collapses them to a vertex.  In other words, any two vertices $m$ and $m'$ are merged when there exists a path from $m$ to $m'$ and a path from $m'$ to $m$.  This technique also handles cases like those derived from the \emph{Condorcet paradox}.

Finally, with the help of the collapsed graph, all payload messages defined by the cut are \op{of-delivered} in a fair order: First, all vertices without any incoming edges are selected.
Secondly, these vertices are sorted in a deterministic way and the corresponding payloads are \op{of-delivered} one after the other. Then the processed vertices are removed from the graph and another iteration through the graph starts.
As soon as there are no vertices left, i.e., all payload messages are \op{of-delivered}, the protocol proceeds to the next round.

Note that cycles may also extend beyond the cut, as shown by Kelkar \etal~\cite{DBLP:journals/iacr/KelkarDLJK21}.  Therefore, the algorithm holds back payload messages and does not \op{of-deliver} them while they may still become part of a longer cycle. This is ensured by counting how many times a message appears in $\msgs$ up to the cut. In particular, let $C[m]$ count this number for a message $m$. We require that any message is only \op{of-delivered} when $C[m] \geq \frac{n+f-\kappa}{2}$, i.e., after $m$ appears in $\msgs$ often enough such that it cannot become part of a cycle later or already be in a cycle that will grow later, e.g., through payloads that arrive after the cut.

\begin{example}\label{ex:1}
  Let us consider a system of $n=4$ processes, of which three ($p_1$,
  $p_2$, and $p_3$) are correct and one ($p_4$) is faulty ($f=1$).  We fix
  the order-fairness parameter $\kappa = 0$, the notion is trivially
  satisfied for higher values of $\kappa$. Every correct process
  $\op{of-broadcasts}$ three messages $m_a$, $m_b$, and $m_c$, in an order that forms a Condorcet cycle. The Byzantine process $p_4$ does not $\op{of-broadcast}$.
  Suppose all messages have been $\op{bcch-delivered}$ in round~$r$ to all
  correct processes, as shown in Figure~\ref{fig:a}. Then the protocol
  obtains the cut~$c = \begin{bmatrix}3&2&1&0\end{bmatrix}$.

  \begin{figure}
    \begin{center}
      \includegraphics[width=0.8\linewidth]{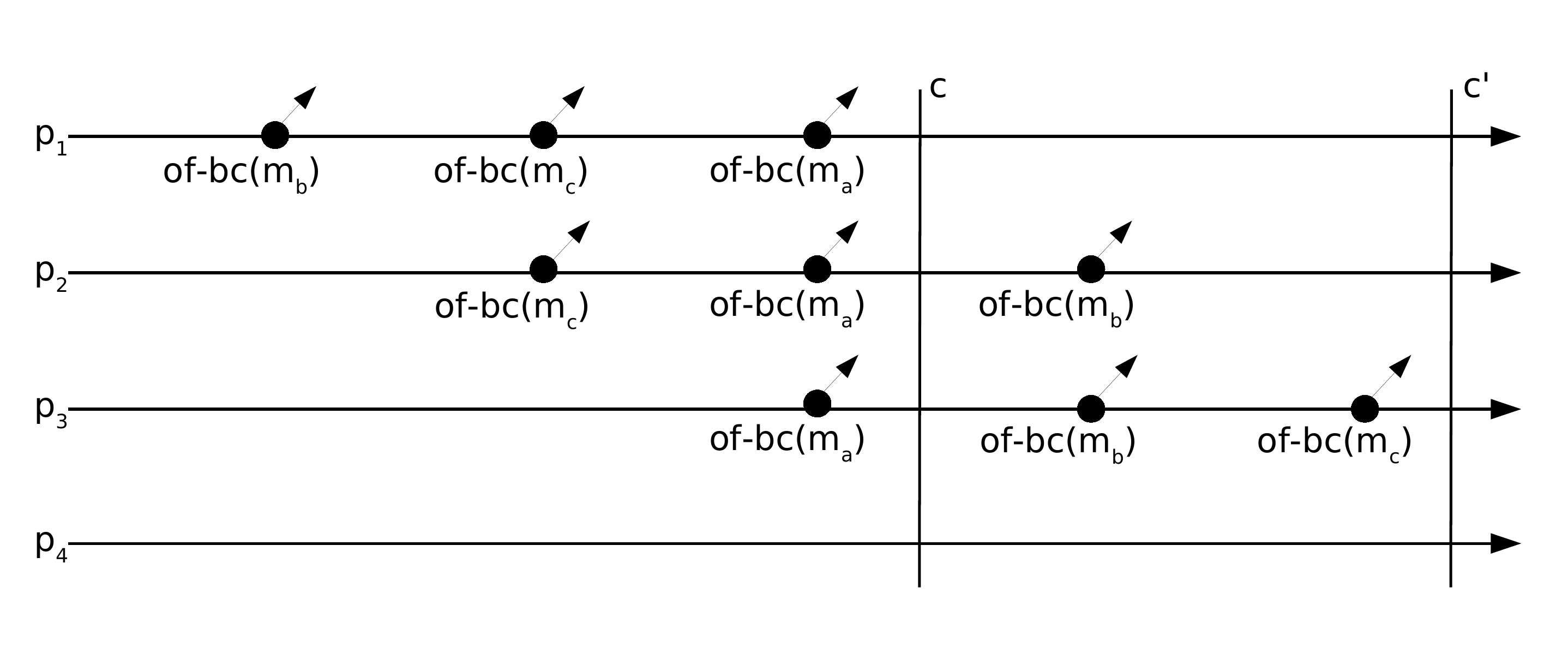}
      \vspace*{-6ex}
    \end{center}
    \caption{The execution of Example~\ref{ex:1}, in which three correct
      processes $p_1, p_2, p_3$ \op{of-broadcast} messages that form a
      cycle, which makes it impossible to sort them in a fair order.  After
      $m_a$, and after the protocol has computed cut~$c$, an unbounded
      number of additional payloads might follow (see text).}
    \label{fig:a}
  \end{figure}

From Algorithm~\ref{alg:alg1}-\ref{alg:alg2} (L\ref{ordermatrix}), the
matrix $M$ and the corresponding graph (L\ref{edges}) are

\hfill
\begin{minipage}[c]{0.45\linewidth}
  \begin{eqnarray*}
    M = \begin{bmatrix}
      0 & 0 & 0\\
      1 & 0 & 1 \\
      2 & 0 & 0 
    \end{bmatrix} 
  \end{eqnarray*}
\end{minipage}
\hfill
\begin{minipage}[c]{0.45\linewidth}
  \includegraphics[height=3.5cm]{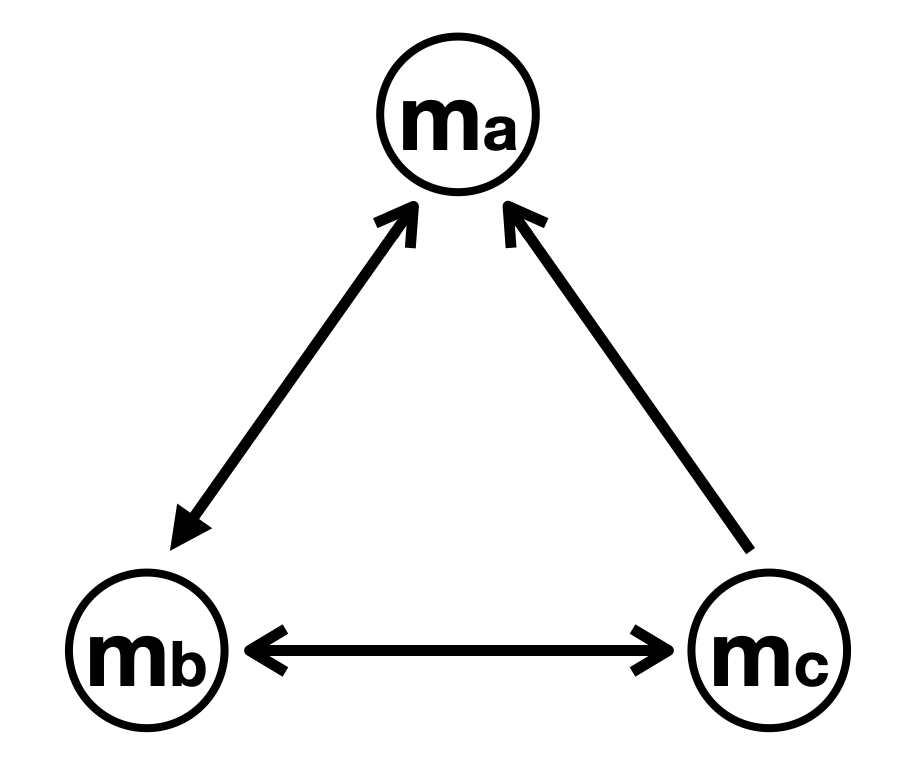}
\end{minipage}
\hfill

Notice that arbitrarily many payload messages that are \op{of-broadcast}
immediately after $m_a$ by $p_1$--$p_3$ might follow and arrive only in a
future round, after cut~$c$.  The protocol cannot know this yet and must
therefore postpone \op{of-delivery} of $m_a$.  As captured by the condition
that $C[m_b] = 1 < \frac{n+f-\kappa}{2}$, no payload message is
$\op{of-delivered}$ in this round.
The protocol continues with another round $r'$ obtaining a cut $c'$,
cf.~Figure~\ref{fig:a}.  Then the matrix $M$ and the graph become

\hfill
\begin{minipage}[c]{0.45\linewidth}
  \begin{eqnarray*}
    M = \begin{bmatrix}
      0 & 2 & 1\\
      1 & 0 & 2 \\
      2 & 1 & 0 
    \end{bmatrix} 
  \end{eqnarray*}
\end{minipage}
\hfill
\begin{minipage}[c]{0.45\linewidth}
  \includegraphics[height=3.5cm]{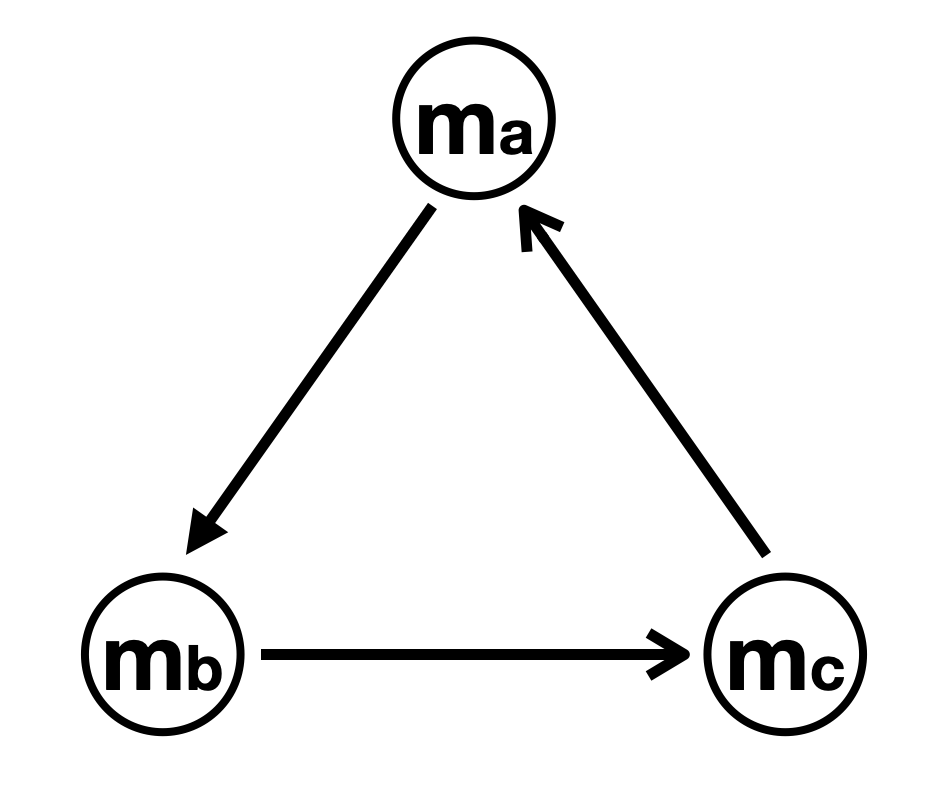}
\end{minipage}
\hfill

At this point, the protocol $\op{of-delivers}$ $\{m_a, m_b, m_c\}$
together, from a collapsed vertex, because now
$C[m_i]=3\geq \frac{n+f-\kappa}{2}$ for $i \in \{a,b,c\}$.
\end{example}

\subsection{Implementation}
\label{subsection:details}

Algorithm~\ref{alg:alg1}--\ref{alg:alg2} shows the \emph{quick order-fair atomic broadcast protocol} for a process~$p_i$.  The protocol proceeds in rounds, maintains a round counter~$r$ (L\ref{r}), and uses a boolean variable \inround, which indicates whether the consensus phase of a round is executing~(L\ref{inround}).

Every process maintains two hash maps: \msgs (L\ref{msgs}) and \vc (L\ref{vc}), in which process identifiers serve as keys. Hash map \msgs contains ordered lists of \op{bcch-delivered} payload from each process in the system. Variable \vc is a vector clock counting how many payload messages were \op{bcch-delivered} from each process.

\paragraph{Rounds.}

In each round, a matrix~$L$ (L\ref{matrix_L}) and a list~$\Sigma$ (L\ref{list_Sigma}) are constructed as inputs for consensus. The matrix~$L$ will consist of vector clocks from the processes and $\Sigma$ will contain the signatures of the processes.
Additionally, every process maintains a list of integers called \cut (L\ref{cutVar}) that are calculated in every round. This cut represents an index for every list in \msgs to determine the payload to be used for creating the fair order. Initially, all values are zero.
Finally, all \op{of-delivered} payload messages are included in a set \delivered (L\ref{del}), to prevent a repeated delivery in future rounds.

The protocol starts when a client submits a payload message~$m$ using an \op{of-broadcast}($m$) event. BCCH then broadcasts $m$ to all processes in the network (L\ref{bcch-broadcast}).  When $m$ with label~$l$ from process~\ppj is \op{\bcch-delivered} (L\ref{bcch-deliver}), the vector clock \vc for process~\ppj is incremented.  The attached label~$l$ is not used by the algorithm and only serves to define that all correct processes \op{\bcch-deliver} the same payload following Definition~\ref{def:bcch}.  Additionally, payload~$m$ is appended to the list $\msgs[j]$ using an operation \op{append}($m$) (L\ref{append}).

When the length of $p_j$'s list in \msgs exceeds the \cut value for~$p_j$, new payloads may have arrived that should be ordered (L\ref{start_round}).  This tells the protocol to initiate a new round.  A new round could also be triggered later, as described in the remarks at the end of this section.

The first step of round~$r$ is to set the flag \inround.  Secondly, the protocol digitally signs the vector clock~\vc and obtains a signature~$\sigma$.
The values~$r$, $\sigma$, and \vc are then sent in a \str{status} message to all processes (L\ref{round_inround}--L\ref{sendstatus}).
When process~\ppi receives a \str{status} message from \ppj, it validates the contained signature~$\sigma'$ using \op{verify}($j,\vc',\sigma'$) (L\ref{verify}). An additional security check is made by comparing the locally stored round number~$r$ with the round number~$r'$ from the message.
If both conditions hold, the vector clock~$\vc'$ is stored as row $j$ in matrix~$L$ (L\ref{L}) and $\sigma'$ is stored in list~$\Sigma$ at index~$j$ (L\ref{Sigma}).

\paragraph{Defining a cut.}

As soon as \ppi has received $n-f$ valid \str{status}-messages (L\ref{waitforsignatures}), it invokes consensus (VBC, L\ref{v-propose}) for the round through \op{vbc-propose} with proposal $(L, \Sigma)$.  The predicate of VBC checks that a proposal consists of a matrix $L$ and a vector~$\Sigma$ such that for at least $n-f$ values~$j$, the entry $\Sigma[j]$ is a valid signature on row~$j$ of $L$.
When the VBC protocol subsequently decides, it outputs a common matrix $L'$ of vector clocks and a list $\Sigma'$ of signatures (L\ref{decide}). The process then uses $L'$ to calculate the cut, where $\cut[j]$ is the largest value $s$ such that at least $f+1$ elements in column $j$ in $L'$ are bigger or equal than $s$ (L\ref{cut}). In other words, $\cut[j]$ represents how many payload messages from \ppj were \op{bcch-delivered} by enough processes. This value is used as index into $\msgs[j]$ to determine the payloads that will be considered for creating the order in this round.

The algorithm then makes sure that all processes will hold at least all those payloads in \msgs that are defined by \cut.  Each process detects missing payload messages from sender \ppj from any difference between $\vc[j]$ and $\cut[j]$ (L\ref{check_missing}); if there are any, the process broadcasts a \str{missing}-message to all others.
When another process receives such a request from \ppj and already has the requested payloads in \msgs, it extracts them into a variable~\emph{resend} (L\ref{getmsgs}). More precisely,
it extracts a proof from the BCCH primitive with which any other process can verify that the payload from this particular sender is genuine.  This is done by invoking $\op{\bcch-create-proof}(\resend)$ (L\ref{createproof}); the messages and the proof are then sent in a \str{resend}-message to the requesting process~\ppj (L\ref{resend_missing}).

When process~\ppi receives a \str{resend}-message with a missing payload from $p_k$, it first verifies the provided proof $s'$ from the message by invoking a $\op{\bcch-verify-proof}(s')$ function (L\ref{verify_resend}). If the proof is valid, \ppi 
extracts (L\ref{bcch-get-msgs}) the payload messages through $\op{\bcch-get-messages}(s')$, appends them to $\msgs[k]$, and increments $\vc[k]$ accordingly.  The process repeats this until \msgs contains all payloads included in the cut.

\begin{algo*}
  \vbox{
  \small
  \begin{numbertabbing}
    xxxx\=xxxx\=xxxx\=xxxx\=xxxx\=xxxx\=MMMMMMMMMMMMMMMMMMM\=\kill
    \textbf{State}\\
    \> \(r \gets 1\): current round \label{r} \\
    \> \(\inround \gets \false\) \label{inround}\\
    \> $\var{msgs} \gets \emptyhashmap: \text{HashMap}\bigl[\{1, ..., n\} \to \emptyhashmap\bigr]$: array of ordered lists of \op{\bcch-delivered} messages \label{msgs}\\
    \> $\vc \gets \emptyhashmap: \text{HashMap}\bigl[\{1, ..., n\} \to \BN\bigr]$: array of counters for \bcch-delivered messages \label{vc}\\
    \> \(L \gets [0]^{n\times n}\): matrix of logical timestamps, constructed from $n$ vector clocks \label{matrix_L}\\
    \> \(\Sigma \gets \emptyhashmap^n\): list of signatures from \str{status} messages \label{list_Sigma}\\
    \> \(\cut \gets [0]^n\): the cut decided for the round \label{cutVar}\\
    \> \(\delivered \gets \emptyset\): set of delivered messages \label{del}\\
    
    \\
    \textbf{Initialization}\\
    \> Byzantine FIFO consistent broadcast channel (\bcch) \newline \label{}\\
    \\
    \textbf{upon} $\op{of-broadcast}(m)$ \textbf{do} \label{}\\
    \> \(\op{\bcch-broadcast}(m)\) \label{bcch-broadcast}\\ 
    \\
    \textbf{upon} \(\op{\bcch-deliver}(\ppj, l, m)\) \textbf{do} \label{bcch-deliver}\\
    \> \(\vc[j] \gets \vc[j] +1 \) \label{incrementvc1}\\
    \> \(\msgs[j]\).\op{append}($m$) \label{append}\\
    \\
    \textbf{upon} exists $j$ \textbf{such that} $\op{len}(\msgs[j])>\cut[j]\land \neg \inround$ \textbf{do} \label{start_round}
    \` // perhaps waiting longer\\
    \> $\inround \gets \true$ \label{round_inround}\\
    \> \(\sigma \gets \op{sign}(i, \vc) \) \label{sign}\\
    \> send message \([\str{status}, r, \vc, \sigma] \) to all \(\ppj \in \CP\) \label{sendstatus}\\ 
    \\
    \textbf{upon} receiving message \([\str{status}, r', \vc', \sigma'] \) from \(\ppj\) \textbf{such that} $r' = r \land \op{verify}(j, \vc', \sigma')$ \textbf{do} \label{verify}\\
    \> \(L[j] \gets \vc'\) \label{L}\\
    \> \(\Sigma[j] \gets \sigma'\) \label{Sigma}\\
    \\
    \textbf{upon} \( \bigl \vert\{p_j \in \CP \mid \Sigma[j] \neq \perp \} \bigr \vert \geq n-f\) \textbf{do} \label{waitforsignatures}\\
    \> \(\op{vbc-propose}\bigl((L, \Sigma)\bigr)\) for validated Byzantine consensus in round \(r\) \label{v-propose}\\
    \> $\Sigma \gets \emptyhashmap^n$ \label{reset_sigma}\\
    \\
    \textbf{upon} \(\op{vbc-decide}\bigl((L', \Sigma')\bigr)\) in round $r$ \textbf{do} \` // calculate the cut \label{decide}\\
    \> \textbf{for} $j \in \{1,\dots,n\}$ \textbf{do} \` // for each row in \(L'\) \label{}\\
    \>\> // $\cut[j]$ is the largest $s$ such that at least $f+1$ elements in column $j$ in $L'$ are at least $s$ \label{}\\
    \>\> $cut[j] \gets \op{max} \bigl\{s \mid \{k \mid \bigl| \{L'[k][j] \geq s\}\bigr| > f\}\bigr\}$ \label{cut}\\[1ex]
    \>\textbf{for} $j \in \{1,\dots,n\}$ \textbf{do} \` // check for missing messages \label{}\\
    \>\>\textbf{if} $\vc[j] < cut[j]$ \textbf{then}  \` // some messages that are \op{\bcch-broadcast} by \ppj are missing\label{check_missing}\\
    \>\>\> send message \([\str{missing}, r, j, \vc[j]] \) to all $p_k \in \CP $ \label{send_missing}\\
    \\
    \textbf{upon} receiving message $[\str{missing}, r', k, \ind]$ from \ppj \textbf{such that} $r' = r$ \textbf{do} \label{receive_missing}\\
    \>\textbf{if} $\vc[k] \geq \cut[k]$ \textbf{then}\label{check-vc-cut}\\
    \>\> $\resend \gets \msgs[k][\ind \dots \cut[k]]$ \` // copy messages from $p_k$ \label{getmsgs} \\
    \>\> $s \gets \op{\bcch-create-proof}(\resend)$ \label{createproof}\\
    \>\> send message \([\str{resend}, r, k, s] \) to \(p_j\) \` // send missing messages to $p_j$ \label{resend_missing}\\
    \\
    \textbf{upon} receiving message $[\str{resend}, r', k', s']$ from \ppj \textbf{such that} $r' = r \land \op{len}(\msgs[k]) < \cut[k]$ \textbf{do} \label{receive_resend}\\
    \> \textbf{if} $\op{bcch-verify-proof}(s')$ \textbf{then} \label{verify_resend}\\
    \>\> \(\vc[k] \gets \vc[k] + \op{\bcch-get-length}(s') \) \label{bcch-get-len}\\
    \>\> \(\msgs[k].\op{append}(\op{\bcch-get-messages}(s')) \) \label{bcch-get-msgs}
  \end{numbertabbing}
  }
  \caption{Quick order-fair atomic broadcast (code for $p_i$).}
  \label{alg:alg1}
\end{algo*}

\begin{algo*}
  \vbox{
  \small
  \begin{numbertabbing}
  xxxx\=xxxx\=xxxx\=xxxx\=xxxx\=xxxx\=MMMMMMMMMMMMMMMMMMM\=\kill
  \textbf{upon} $\op{len}(\msgs[j]) \geq \cut[j]$ for all $j \in \{1, \dots, n\}$  \textbf{do} \label{}\\
  \> $V \gets \Bigl(\bigcup_{j \in \{1, \dots, n\}} \msgs\bigl[j\bigr]\bigl[1 \dots \cut[k]\bigr] \Bigr) \setminus \delivered$ \label{filter_delivered}\\
  \> $M \gets \emptyhashmap: \text{HashMap}\bigl[\CM \times \CM \to \BN\bigr]$
  \` // counts in how many \var{msgs} arrays $m$ occurs before $m'$ \label{constructM}\\
  \> $C \gets \emptyhashmap: \text{HashMap}\bigl[\mathcal{M} \to \BN\bigr]$
  \` // counts how many times $m$ appears in \var{msgs} arrays \label{C}\\
  \> \textbf{for} \(m, m' \in V\) \textbf{do} \label{}\\
  \>\> \(M[m][m'] \gets \Bigl\vert \bigl\{j \in \{1, \dots, n\} \,\big|\, m \text{\ appears before\ } m' \text{\ in\ } \msgs\bigl[j\bigr]\bigl[1 \dots \cut[k]\bigr] \bigr\} \Bigr\vert\) \label{ordermatrix}\\
  \>\> $C[m] \gets \Bigl| \bigl\{p_j \,\big|\, m \in \msgs\bigl[j\bigr]\bigl[1 \dots \cut[k]\bigr] \bigr\}\Bigr|$\label{constructC}\\[1ex]
  \> \(E \gets \Bigl\{(m, m') \,\Big|\, \max\bigl\{M[m][m'], n-f-M[m'][m]\bigr\} > M[m'][m] -f + \kappa \Bigr\} \) \` // add all edges \label{edges}\\
  \> \( H \gets (V, E) \)
  \` // $(V,E) = G$ \label{collapsed_graph} \\
  \> \textbf{while} $H$ contains some strongly connected subgraph $\overline{H} = (\overline{W}, \overline{F}) \subseteq H$ \textbf{do} \label{scc}\\
  \>\> \(H \gets H / \overline{F} \) \` // collapse vertices in $\overline{W}$ into a single vertex $\bar{w}$ via edge contraction \label{edgecontraction}\\[1ex]
  \> // $H = (W, F)$ \label{}\\
  \> \textbf{while} $\exists w \in  \op{sort}(W): \op{indegree}(w) = 0 \land \op{stable}(w)$ \textbf{do}
  \` // in deterministic order \label{startsorting}\\
  \>\> $\op{of-deliver}(\op{flatten}(w))$
  \` // $w$ may be a message or a (recursive) set of sets of messages \label{of-deliver}\\
  \>\> $\delivered \gets \delivered \cup \op{flatten}(w)$ \` // keep track of delivered messages \label{addtodelivered}\\
  \>\> \(W \gets W \setminus \{w\} \) \label{removefromgraph}\\[1ex]
  \> $L \gets [0]^{n\times n}$ \label{reset}\\
  \> $\inround \gets \false$ \label{}\\
  \> $r \gets r + 1$\` // move to the next round \label{round}\\
  \\
  \textbf{function} $\op{stable}(w)$
  \` // $w$ may be a message from \CM or a (recursive) set of sets of messages \label{funcstable}\\
  \> \textbf{return} $\bigl(w \in \CM \land C[w] \geq \frac{n+f-\kappa}{2} \bigr)
  \lor \bigwedge_{w' \in w : w' \not\in \CM} \op{stable}(w')$ \label{}\\
  \\
  \textbf{function} $\op{flatten}(w)$
  \` // $w$ may be a message from \CM or a (recursive) set of sets of messages \label{funcflatten}\\
  \> \textbf{return} $\{m \in w \mid m \in \CM\} \cup \bigcup_{w' \in w : w' \not\in \CM} \op{flatten}(w')$ \label{}\\
  \end{numbertabbing}
  }
  \caption{Quick order-fair atomic broadcast (code for $p_i$).}
  \label{alg:alg2}
\end{algo*}

\paragraph{Ordering messages.}

At this point, every process stores all payloads \msgs that have been \op{bcch-delivered} up to the cut.  The remaining operations of the round are deterministic and executed by all processes independently.
The next step is to construct the directed \emph{dependency graph}~$G$ that expresses the constraints on the fair order of the payload messages. Vertices ($V$) in $G$ represent payload messages that may be \op{of-delivered} and edges ($E$) in $G$ express constraints on the order among these payloads.
First, all messages within the cut that are not yet delivered are added as vertices to the set $V$ (L~\ref{filter_delivered}).

Then, for each pair of messages $m$ and $m'$ in $V$, the algorithm constructs $M$~(L\ref{ordermatrix}) such that $M[m][m']$ counts how many times a payload $m$ appears before payload $m'$ in the cut. In the same loop, the algorithm counts how many times message $m$ appears within the cut and stores this result in array~$C$~(L\ref{constructC}). Finally, all entries $M[m][m']$ and $M[m'][m]$ are compared and if condition $\op{max}\{M[m][m'], n-f-M[m'][m] \} > M[m'][m]-f+\kappa$ holds, then a directed edge from $m$ to $m'$ is added (L\ref{edges}).  This edge indicates that $m$ must \emph{not} be ordered \emph{after} $m'$, i.e., that $m$ is \op{of-delivered} before $m'$ or together with~$m'$.

Any payloads that cannot be ordered with respect to each other now correspond to strongly connected components of $G$. A strongly connected component is a subgraph, which for each pair of vertices $m$ and $m'$ contains a path from $m$ to $m'$ and one from $m'$ to $m$. In the next step, a graph $H=(W,F)$ is created and all strongly connected components in $H$ are repeatedly collapsed until $H$ contains no more cycles. This is done by contracting the edges in each connected component and merging all its vertices (L\ref{collapsed_graph}--L\ref{edgecontraction}).

The algorithm further considers all vertices $w$ without incoming edges and which satisfy condition $C[m] \geq \frac{n+f-\kappa}{2}$, checked in function $\op{stable}(w)$ (L~\ref{funcstable}). All such $w$ will be sorted in a deterministic way (L~\ref{startsorting}). Notice that $w$ may correspond to a message from \CM or a recursive set of sets of messages. Therefore function $\op{flatten}(w)$ (L~\ref{funcflatten}) is used to extract payload messages and \op{of-deliver} them (L~\ref{of-deliver}). All \op{of-delivered} payload messages are added to \delivered (L\ref{addtodelivered} to prevent a repeated processing. Finally, $w$ is removed from $H$ (L\ref{removefromgraph}), and a next pass of extracting vertices with no incoming edge follows.  This is repeated until all vertices have been processed and \op{of-delivered}.

The algorithm then initializes $L$, sets \inround to $\false$, increments the round number~$r$, and starts the next round (L\ref{reset}-L\ref{round}).

\paragraph{Remarks.}

The condition for starting a round in L\ref{start_round} only waits until
\emph{one} single payload exists in \msgs that was not considered
before. This is necessary for liveness but not very efficient.  This number
can be increased such that a new round starts only after $K = \Theta(n)$
new payload messages have arrived. Note that this threshold affects the
amortized message and bit complexities that are considered in
Section~\ref{subsection:complexity}.

Recall that our model assumes that every correct process \op{of-broadcasts}
all payload messages.  For simplicity, though, our validity property has
been formulated only for executions without faulty processes.  It could be
strengthened so that it holds for all executions, in which the processes do
not \op{of-broadcast} an unbounded number of them that form a Condorcet
cycle.

The protocol can also be changed to satisfy the even stronger liveness
property of Kelkar \etal~\cite{DBLP:journals/iacr/KelkarDLJK21}, which the
Themis protocol satisfies.  To deal with Condorcet cycles of unbounded
length, one would modify the interface of order-fair broadcast so that it
additionally outputs \op{of-startblock} and \op{of-endblock} events that
carry no parameters.  Furthermore, \op{of-deliver} would only output single
payload messages from~\CM.  An output ``block'' then consists of all
payloads that are \op{of-delivered} between a \op{of-startblock} event and
the subsequent \op{of-endblock} event.  However, long cycles occur very
infrequently in realistic scenarios, as shown by Kelkar
\etal~\cite{DBLP:journals/iacr/KelkarDLJK21}.

If consensus is not ``black box'' and treated in a modular way, more
efficient variations of this protocol become possible.  In particular, the
ordering rounds may be integrated with a leader-based Byzantine consensus
protocol~\cite{DBLP:books/daglib/0025983}.  This implies that multiple
leaders in successive consensus rounds (or ``epochs'') may be needed to
agree on the cut of one ordering round.  The Themis
protocol~\cite{DBLP:journals/iacr/KelkarDLJK21} adopts this pattern.

The protocol satisfies another natural property, which has not been made
explicit before in the literature, but is achieved by several existing
protocols~\cite{DBLP:conf/crypto/Kelkar0GJ20, DBLP:conf/osdi/ZhangSCZA20,
  DBLP:journals/iacr/KelkarDLJK21}, not only by quick order-fair broadcast.
Consider an execution in which the correct processes \op{of-broadcast}
messages $m_1, \dots, m_l$ such that $b(m_i, m_j) > 2f + \kappa$ for
$i = 1, \dots, l$ and $j = i+1, \dots, l$ and there are no further messages
\op{of-broadcast} that might include $m_1, \dots m_l$ in a cycle: Then
$m_i$ is actually \op{of-delivered} before $m_j$.  Note that differential
order fairness is a safety condition and would not prevent that
$m_1, \dots, m_l$ are \op{of-delivered} jointly in one set.

\subsection{Complexity}
\label{subsection:complexity}

In this section, we analyze the complexity of the \emph{quick order-fair atomic broadcast protocol}. We use two measures: message complexity and communication (bit) complexity. Moreover, we compare our results with existing algorithms from the literature.

\paragraph{Message complexity.}

If the Byzantine FIFO consistent broadcast channel (BCCH) is implemented using ``echo broadcast''~\cite{DBLP:conf/ccs/Reiter94}, it takes $O(n)$ protocol messages per payload message.  Since more than $f$ processes \op{of-broadcast} each payload message and $f$ is proportional to $n$, the overall message complexity of BCCH is~$O(n^2)$.  Under high load, batching could be used to reduce the number of messages incurred by BCCH.  In the protocol itself, every process sends $O(n)$ \str{status}, \str{missing}, and \str{resend} messages, which also amounts to $O(n^2)$ messages.

The cost of validated Byzantine consensus (VBC) depends on the assumptions used for implementing it.  In the asynchronous model, optimal protocols~\cite{DBLP:conf/podc/AbrahamMS19,DBLP:conf/podc/LuL0W20} achieve $O(n^2)$ messages on average.  Assuming that $K$ new payload messages are delivered in each round, this becomes $O(\frac{n^2}{K})$ per payload. Choosing $K = \Omega(n)$ reduces the amortized cost of consensus to $O(n)$ messages per payload message. Note that when using an implementation of VBC with complexity $O(n^3)$, as the algorithm of Cachin \etal~\cite{DBLP:conf/crypto/CachinKPS01}, we can choose $K$ proportional to $n$ and may again obtain expected amortized message complexity~$O(n^2)$.

With a partially synchronous consensus protocol according to Section~\ref{sec:vbc}, VBC uses $O(n)$ messages in the best case and $O(n^2)$ messages in the worst case.  The total amortized cost of quick order-fair atomic broadcast per payload, therefore, is also $O(n^2)$ messages in this implementation.

\paragraph{Communication (bit) complexity.}

If digital signatures are of length $\lambda$ and payload messages are at most $L$ bits, the bit complexity of BCCH for one sender is $O(nL + n^2 \lambda)$, and since we assume that $O(n)$ processes broadcast each message, this becomes $O(n^2 L + n^3 \lambda)$. Optimal asynchronous VBC protocols~\cite{DBLP:conf/podc/AbrahamMS19,DBLP:conf/podc/LuL0W20} have $O(nL + n^2 \lambda)$ expected communication cost, for their payload length~$L$.  Since the proposals for VBC are $n \times n$ matrices, the bit complexity of this phase  is $O(n^3 + n^2 \lambda)$. Assuming that $K$ is proportional to $n$, the amortized bit complexity of VBC per payload is $O(n^2 + n\lambda)$. From this, it follows that the amortized bit complexity of the algorithm per payload message is~$O(n^2L + n^3\lambda)$.

\paragraph{Discussion.}

Table~\ref{table:complexity} gives an overview of message complexities of algorithms with different notions for fair payload message ordering. We compare our \emph{quick order-fair atomic broadcast} with the algorithms introduced by Kelkar \etal~\cite{DBLP:conf/crypto/Kelkar0GJ20} and Zhang \etal~\cite{DBLP:conf/osdi/ZhangSCZA20}. We leave out from the overview the protocol by Kursawe~\cite{DBLP:conf/aft/Kursawe20} since it has a completely different approach for solving fair payload message ordering. 

The asynchronous Aequitas protocol~\cite[Sec.~7]{DBLP:conf/crypto/Kelkar0GJ20} provides fair order using a \emph{FIFO Broadcast primitive}, implemented by \emph{OARcast} of Ho \etal~\cite{DBLP:conf/opodis/HoDR07}.  The implementation of OARcast described there uses $n$ \emph{ARcasts}~\cite{DBLP:conf/opodis/HoDR07} for each payload, where one ARcast causes $\Theta(n^2)$ network messages.  Since Aequitas requires that every correct process broadcasts each payload, the total complexity increases by another factor of~$n$.  Thus, each payload message incurs a cost of $\Theta(n^4)$ messages in the gossip phase.  Moreover, one instance of \emph{set agreement} is executed for each payload, and each one of them calls $n$ binary consensus protocols.  Therefore Aequitas uses $\Omega(n^4)$ messages for delivering one payload, which exceeds the cost of quick order-fair broadcast at least by the factor~$n^2$.

Ordering linearizability~\cite{DBLP:conf/osdi/ZhangSCZA20} is defined using
a logical order of events observed on each process.  Its implementation in
the Pomp\=e protocol, however, appears to require synchronized clocks in
the sense of knowing bounds on differences between local clocks.
Hence, the complexity of Pomp\=e cannot be compared to that of asynchronous
protocols for order fairness.  Irrespective of this difference, its cost is
$O(n^2)$ messages and one instance of Byzantine consensus per payload
message.  The communication complexity of this protocol is $O(n^3 L)$ since
each process broadcasts a \str{sequence}-message to all others with
contents of length~$O(nL)$.

Themis~\cite{DBLP:journals/iacr/KelkarDLJK21} relies strongly on a leader
$p_\ell$ to construct a fair order.  If $p_\ell$ does not perform its task
timely, the protocol may switch to another leader, similarly to existing
leader-based protocols.  For assessing the complexity of Themis here, we
consider the optimistic case, but note that the complexities stated for
some other protocols, in particular for the quick order-fair broadcast, do
not depend on timely leaders.

Themis lets all processes send their local orderings to~$p_\ell$ first.
Suppose these consist of approximately $K = \Theta(n)$ payload messages
each.  Then $p_\ell$ constructs a graph $G$ on these and sends $G$ and some
justification information to all processes.  They maintain local graphs,
update them in response, and potentially output some payload messages.
This incurs a cost of $O(n)$ messages.  Since $G$ contains $K$ nodes and,
in general, $O(K^2)$ edges, the average communication complexity is
$O(n^2 + nL)$ in the best case.
  
\begin{table}[ht!]
  \centering
  \begin{tabular}{|l|l|c|c|} 
    \hline
    Notion & Algorithm & Avg.\ messages & Avg.\ communication \\ 
    \hline
    Block-Order-Fairness~\cite{DBLP:conf/crypto/Kelkar0GJ20} & Async.\ Aequitas~\cite{DBLP:conf/crypto/Kelkar0GJ20} & $O(n^4)$ & $O(n^4 L)$ \\ 
    Ordering Linearizability~\cite{DBLP:conf/osdi/ZhangSCZA20} & Pomp\=e$^*$~\cite{DBLP:conf/osdi/ZhangSCZA20} & $O(n^2)$ & $O(n^3 L)$ \\ 
    Block-Order-Fairness~\cite{DBLP:journals/iacr/KelkarDLJK21} & Themis~\cite{DBLP:journals/iacr/KelkarDLJK21} & $O(n)$ & $O(n^2 + nL)$ \\ 
    Differential Order Fairness & Quick o.-f. broadcast & $O(n^2)$& $O(n^2 L + n^3 \lambda)$\\
    \hline
  \end{tabular}
  \caption{Overview of different notions for fair message ordering and
    corresponding algorithms, with their expected message and communication
    complexities.  The summary assumes $L \geq \lambda$.
    ($^*\,$The Pomp\=e protocol requires synchronized clocks.)}
  \label{table:complexity}
\end{table}

\section{Analysis}

In this section, we show that the \emph{quick order-fair atomic broadcast protocol} in Algorithm~\ref{alg:alg1}--\ref{alg:alg2} implements $\kappa$-differentially order-fair atomic broadcast. The properties to be satisfied are~(Definition~\ref{def:kappa-of}): \emph{no duplication, agreement, total order, strong validity and $\kappa$-differential order fairness}.

\begin{lemma}\label{analysis:noduplication}
  No message is \op{of-delivered} more than once in
  Algorithm~\ref{alg:alg1}--\ref{alg:alg2}.
\end{lemma} 
\begin{proof}
  The check in L\ref{filter_delivered} of the protocol implementation
  ensures that no payload message is \op{of-delivered} more than once. In
  the step when the protocol creates graph vertices, payload messages that
  are already contained in variable \delivered are filtered out. Those
  messages will not be included in the graph and cannot be
  \op{of-delivered} again. Note that even in the case when a payload
  message~$m$ is \op{bcch-delivered} multiple times, because of filtering
  in L\ref{filter_delivered}, it is not possible that $m$ is
  \op{of-delivered} more than once.
\end{proof}

\begin{lemma}\label{analysis:agreement}
  In Algorithm~\ref{alg:alg1}--\ref{alg:alg2}, if a message $m$ is
  \op{of-delivered} by some correct process, then $m$ is eventually
  \op{of-delivered} by every correct process.
\end{lemma}
\begin{proof}
  Suppose that a payload message $m$ is \op{of-delivered} by some correct
  process~\ppi in round $r$. Following the protocol steps, in round $r$ all
  correct processes decide on the same $L'$ (L\ref{decide}). This is
  guaranteed by the \emph{agreement} property of the validated Byzantine
  consensus because no two correct processes decide differently. Since the
  matrix $L'$ is used to construct the cut deterministically, all correct
  processes construct the same \cut~(L\ref{cut}).
  
  We can then distinguish two cases: In the first case, all correct
  processes have already \op{bcch-delivered} $m$ and store it in
  $\msgs$. In the second case, there are some correct processes that have
  never heard of $m$ simply because of some delays in the network.  Then
  these correct processes send a \str{missing}-message to all processes,
  requesting the delivery of their missing payloads.  Since every message
  included in the cut was announced by $f+1$ processes, and therefore also
  by at least one correct process, some process will respond with a
  \str{resend}-message containing~$m$.  Once all these messages are
  delivered, all correct processes store $m$ in~\msgs.
  
  In the next step, every correct process builds graph~$G$. Each vertex in
  the graph is constructed deterministically from the same information by
  every process, concretely, from the payload messages in \msgs and
  excluding those that are already in the \delivered set
  (L\ref{filter_delivered}). Following the protocol, every correct process
  will eventually construct the same $G$ and output the same sequence of
  payload messages, also including~$m$.
\end{proof}

\begin{lemma}\label{analysis:totalorder}
  Let $m$ and $m'$ be two messages such that \ppi and \ppj are correct
  processes that \op{of-deliver} $m$ and $m'$. In
  Algorithm~\ref{alg:alg1}--\ref{alg:alg2}, if \ppi \op{of-delivers} $m$
  before $m'$, then \ppj also \op{of-delivers} $m$ before $m'$.
\end{lemma}
\begin{proof}
  Consider two distinct payload messages $m$ and $m'$ and let \ppi and \ppj
  be any two correct processes that \op{of-deliver} both messages. Assume
  that \ppi \op{of-delivers} $m$ before $m'$. If \ppi \op{of-delivers} $m$
  and $m'$ in round $r$, then both messages were included in the cut
  for~\ppi.  Due to the argument used to establish the \emph{agreement}
  property in Lemma~\ref{analysis:agreement}, it must be that $m$ and $m'$
  were also included in the cut for process~\ppj in round~$r$. The rest of
  the protocol, i.e., building a graph and \op{of-delivering} messages is
  deterministic. Therefore, \ppj delivers these two messages in round~$r$
  and also \op{of-delivers} $m$ before $m'$.  Extending this argument over
  all rounds of the protocol, it follows that every correct process
  \op{of-delivers} the same sequence of payload messages.
\end{proof}

\begin{lemma}\label{analysis:validity}
  If all processes are correct and $\op{of-broadcast}$ a finite number of messages in
  Algorithm~\ref{alg:alg1}--\ref{alg:alg2}, then every correct process
  eventually \op{of-delivers} these messages.
\end{lemma}
\begin{proof}
  Let \ppi be some correct process that \op{of-broadcasts} a payload
  message $m$. Due to the \emph{validity} property of the underlying
  Byzantine FIFO consistent broadcast channel, every correct process
  eventually \op{\bcch-delivers} $m$. According to the algorithm, in every
  round $r$ a process \ppi waits for $n-f$ processes to receive signed
  vector clocks to proposes a matrix of logical timestamps $L$ for
  validated Byzantine consensus (L\ref{v-propose}). The \emph{termination}
  property of validated Byzantine consensus guarantees that every correct
  process eventually decides some value and according to the
  \emph{agreement} property, no two correct processes decide
  differently. The resulting common $L'$ allows then each process to
  determine if $m$ is considered in the current round $r$. A message $m$ is
  considered if at least $f+1$ processes have \op{\bcch-delivered} $m$ and
  reported it in their vector clock (L\ref{cut}). Additionally, if $m$ is
  considered in the current round but some process~\ppi has not
  \op{\bcch-delivered} $m$ yet, \ppi will request that other processes send
  it the missing payload message (L\ref{send_missing}). Further, all
  messages in \msgs are added as vertices to the graph~$G$
  (L\ref{filter_delivered}). 
  Moreover, because every process \op{of-broadcasts} a finite number of messages, every possible graph that is created will be finite.
  Since $m$ was \op{of-broadcast} by a correct process
  \ppi, all processes are correct and $\op{of-broadcast}$ a finite number of messages, $m$ will eventually be \op{of-delivered}.

\end{proof}

\begin{lemma}\label{analysis:kappa}
  In Algorithm~\ref{alg:alg1}--\ref{alg:alg2}, if
  $b(m,m') > b(m',m) + 2f + \kappa$, then no correct process
  \op{of-delivers} $m'$ before $m$.
\end{lemma}
\begin{proof}  
  Recall that $b(m,m')$ is the number of correct processes that receive and
  \op{of-broadcast} $m$ before $m'$. Consider any two payload messages $m$
  and $m'$ such that $b(m,m') > b(m',m) + 2f + \kappa$.

  Suppose $m$ and $m'$ are both included in the cut of some round and none
  of them has been \op{of-delivered} yet.  The protocol defines a threshold
  based on $M$ for creating an edge between two vertices.
  Lemma~\ref{lem:M} shows that the condition for differential order
  fairness ensures that $M[m][m'] > M[m'][m] - f + \kappa$ in the protocol,
  where $M[m][m']$ counts how many times $m$ appears before $m'$ in \msgs.
  Moreover, as explained in connection with Lemma~\ref{lem:M}, the
  algorithm extends this condition for adding an edge $(m,m')$ to
  \begin{equation}\label{eq:max-cut}
    \max\big\{M[m][m'], n-f-M[m'][m]\big\} > M[m'][m] -f + \kappa,
  \end{equation}
  in order to cope with particularly small values of $M[m'][m]$.  This may
  be the case when the full relative ordering information about $m$ and
  $m'$, in the sense that $M[m][m'] + M[m'][m] \geq n-f$, is not yet
  available with the cut.  The implementation then adds an edge from $m$ to
  $m'$ to the graph~(L\ref{edges}).  This implies that $m'$ will not be
  \op{of-delivered} before $m$ because the algorithm respects this order by
  traversing the graph starting with vertices that have no incoming edges.
  Therefore, $m$ is either \op{of-delivered} before $m'$ or both messages
  are delivered together, within the same set.  Moreover, observe that
  (\ref{eq:max-cut}) ensures that graph generated by the protocol is
  connected.

  Consider now the case that $m'$ is not included at all in the cut of the
  current round~$r$.  We want to show that for all $m \in V$ of the
  graph~$G$, if $m$ is $\op{of-delivered}$ in round $r$, there cannot be
  such an~$m'$, for which an edge $(m', m)$ would be added at a later round
  and which might therefore violate $\kappa$-order fairness.  Recall that
  an edge $(m', m)$ is added to $G$ in a round whenever (\ref{eq:max-cut})
  holds.
  
  To be more precise, we show that the condition in L\ref{startsorting} and
  the properties of $\op{stable}()$ ensure $\kappa$-differential order
  fairness for such $m$ and $m'$.
  Let $\bar{w}$ be a node of the graph as in L\ref{startsorting}.  If every
  $m$ in $\op{flatten}(\bar{w})$ appears at least $\frac{n+f-\kappa}{2}$
  times in $\var{msgs}$ up to the cut, i.e., satisfies $\op{stable}(m)$, it
  means that no $m'$ (not in the cut) can be ordered before the messages
  in~$G$ of subsequent rounds.  In fact, let $m \in \op{flatten}(\bar{w})$
  be such that $\op{stable}(m) = \true$ and $m'$ not be in the cut.
  Since $\op{stable}(m)$ holds, $C[m] \geq \frac{n+f-\kappa}{2}$ also means
  that $M[m][m'] \geq \frac{n+f-\kappa}{2}$.  Thus,
  \[
    M[m'][m] \leq n - M[m][m']
    \leq n - \frac{n+f-\kappa}{2}
    = \frac{n-f+\kappa}{2}
  \]
  in any future round as well.  But this implies
  $M[m'][m] - M[m][m'] \leq -f + \kappa$, and thus, no edge $(m', m)$ is
  added according to (\ref{eq:max-cut}).  The argument given earlier then
  shows that order fairness is maintained.
  Notice that this takes care of scenarios as in Example~\ref{ex:1} that
  include some message $\bar{m} \in \op{flatten}(\bar{w})$ with
  $\neg\op{stable}(\bar{m})$.  There may exist a further message $m'$ not
  included in the cut such that $m'$ must be ordered not after $\bar{m}$.
\end{proof}

Lemmas~\ref{analysis:noduplication}--\ref{analysis:kappa} directly imply the following theorem, which concludes the analysis of the protocol.

\begin{theorem}\label{thm:algimpl}
  Algorithm~\ref{alg:alg1}--\ref{alg:alg2} implements $\kappa$-differentially order-fair atomic broadcast.
\end{theorem}

\section{Conclusion}

The quick order-fair atomic broadcast protocol guarantees payload message delivery in a differentially fair order. It works both for asynchronous and eventually synchronous networks with optimal resilience, tolerating corruptions of up to one third of the processes. Compared to existing order-fair atomic broadcast protocols, our protocol is considerably more efficient and incurs only quadratic cost in terms of amortized message complexity per delivered payload.

\section*{Acknowledgments}

We thank the anonymous reviewers for helpful suggestions and feedback.
Special thanks go to Mahimna Kelkar, who pointed out a problem in an
earlier version of this paper.

This work has been funded by the Swiss National Science Foundation (SNSF)
under grant agreement Nr\@.~200021\_188443 (Advanced Consensus Protocols).

\bibliography{dblpbibtex, references}

\begin{thebibliography}{10}

\bibitem{DBLP:conf/podc/AbrahamMS19}
I.~Abraham, D.~Malkhi, and A.~Spiegelman, ``Asymptotically optimal validated
  asynchronous byzantine agreement,'' in {\em {PODC}}, pp.~337--346, {ACM},
  2019.

\bibitem{DBLP:conf/icnp/AsayagCGLRTY18}
A.~Asayag, G.~Cohen, I.~Grayevsky, M.~Leshkowitz, O.~Rottenstreich, R.~Tamari,
  and D.~Yakira, ``A fair consensus protocol for transaction ordering,'' in
  {\em {ICNP}}, pp.~55--65, {IEEE} Computer Society, 2018.

\bibitem{baird2016swirlds}
L.~Baird, ``The {Swirlds} hashgraph consensus algorithm: Fair, fast, byzantine
  fault tolerance.'' Swirlds Tech Report, SWIRLDS-TR-2016-01, 2016.

\bibitem{DBLP:books/daglib/0025983}
C.~Cachin, R.~Guerraoui, and L.~E.~T. Rodrigues, {\em Introduction to Reliable
  and Secure Distributed Programming {(2.} ed.)}.
\newblock Springer, 2011.

\bibitem{DBLP:conf/crypto/CachinKPS01}
C.~Cachin, K.~Kursawe, F.~Petzold, and V.~Shoup, ``Secure and efficient
  asynchronous broadcast protocols,'' in {\em {CRYPTO}}, vol.~2139 of {\em
  Lecture Notes in Computer Science}, pp.~524--541, Springer, 2001.

\bibitem{DBLP:journals/tocs/CastroL02}
M.~Castro and B.~Liskov, ``Practical byzantine fault tolerance and proactive
  recovery,'' {\em {ACM} Trans. Comput. Syst.}, vol.~20, no.~4, pp.~398--461,
  2002.

\bibitem{DBLP:conf/sp/DaianGKLZBBJ20}
P.~Daian, S.~Goldfeder, T.~Kell, Y.~Li, X.~Zhao, I.~Bentov, L.~Breidenbach, and
  A.~Juels, ``Flash boys 2.0: Frontrunning in decentralized exchanges, miner
  extractable value, and consensus instability,'' in {\em {IEEE} Symposium on
  Security and Privacy}, pp.~910--927, {IEEE}, 2020.

\bibitem{DBLP:conf/dsn/DuanRZ17}
S.~Duan, M.~K. Reiter, and H.~Zhang, ``Secure causal atomic broadcast,
  revisited,'' in {\em {DSN}}, pp.~61--72, {IEEE} Computer Society, 2017.

\bibitem{DBLP:journals/jacm/DworkLS88}
C.~Dwork, N.~A. Lynch, and L.~J. Stockmeyer, ``Consensus in the presence of
  partial synchrony,'' {\em J. {ACM}}, vol.~35, no.~2, pp.~288--323, 1988.

\bibitem{DBLP:conf/podc/FitziG03}
M.~Fitzi and J.~A. Garay, ``Efficient player-optimal protocols for strong and
  differential consensus,'' in {\em {PODC}}, pp.~211--220, {ACM}, 2003.

\bibitem{DBLP:conf/opodis/HoDR07}
C.~Ho, D.~Dolev, and R.~van Renesse, ``Making distributed applications
  robust,'' in {\em {OPODIS}}, vol.~4878 of {\em Lecture Notes in Computer
  Science}, pp.~232--246, Springer, 2007.

\bibitem{DBLP:journals/iacr/KelkarDK21}
M.~Kelkar, S.~Deb, and S.~Kannan, ``Order-fair consensus in the permissionless
  setting,'' {\em {IACR} Cryptol. ePrint Arch.}, no.~139, 2021.
\newblock \url{https://eprint.iacr.org/2021/139}.

\bibitem{DBLP:journals/iacr/KelkarDLJK21}
M.~Kelkar, S.~Deb, S.~Long, A.~Juels, and S.~Kannan, ``Themis: Fast, strong
  order-fairness in byzantine consensus,'' {\em {IACR} Cryptol. ePrint Arch.},
  no.~1465, 2021.
\newblock \url{https://eprint.iacr.org/2021/1465}.

\bibitem{DBLP:conf/crypto/Kelkar0GJ20}
M.~Kelkar, F.~Zhang, S.~Goldfeder, and A.~Juels, ``Order-fairness for byzantine
  consensus,'' in {\em {CRYPTO} {(3)}}, vol.~12172 of {\em Lecture Notes in
  Computer Science}, pp.~451--480, Springer, 2020.

\bibitem{DBLP:conf/aft/Kursawe20}
K.~Kursawe, ``Wendy, the good little fairness widget: Achieving order fairness
  for blockchains,'' in {\em {AFT}}, pp.~25--36, {ACM}, 2020.

\bibitem{DBLP:conf/podc/LuL0W20}
Y.~Lu, Z.~Lu, Q.~Tang, and G.~Wang, ``Dumbo-mvba: Optimal multi-valued
  validated asynchronous byzantine agreement, revisited,'' in {\em {PODC}},
  pp.~129--138, {ACM}, 2020.

\bibitem{DBLP:conf/ccs/Reiter94}
M.~K. Reiter, ``Secure agreement protocols: Reliable and atomic group multicast
  in rampart,'' in {\em {CCS}}, pp.~68--80, {ACM}, 1994.

\bibitem{DBLP:journals/toplas/ReiterB94}
M.~K. Reiter and K.~P. Birman, ``How to securely replicate services,'' {\em
  {ACM} Trans. Program. Lang. Syst.}, vol.~16, no.~3, pp.~986--1009, 1994.

\bibitem{DBLP:conf/podc/YinMRGA19}
M.~Yin, D.~Malkhi, M.~K. Reiter, G.~Golan{-}Gueta, and I.~Abraham, ``Hotstuff:
  {BFT} consensus with linearity and responsiveness,'' in {\em {PODC}},
  pp.~347--356, {ACM}, 2019.

\bibitem{DBLP:conf/osdi/ZhangSCZA20}
Y.~Zhang, S.~T.~V. Setty, Q.~Chen, L.~Zhou, and L.~Alvisi, ``Byzantine ordered
  consensus without byzantine oligarchy,'' in {\em {OSDI}}, pp.~633--649,
  {USENIX} Association, 2020.

\end{thebibliography}
\bibliographystyle{ieeesort}

\end{document}
